\documentclass{llncs}

\usepackage{graphicx}
\usepackage{bussproofs}
\usepackage{amsmath}
\usepackage{amssymb}
\usepackage{rotating}

\def\SB#1{\textsubscript{{#1}}}
\def\IT#1{\textit{{#1}}}
\def \BF #1{\textbf{{#1}}}
\newenvironment{bprooftree}
{\leavevmode\hbox\bgroup}
{\DisplayProof\egroup}

\title{On the Calculation of Fundamental Groups in Homotopy Type Theory by Means of Computational Paths}

\author{Tiago Mendon\c{c}a Lucena de Veras\inst{1} \and Arthur F. Ramos\inst{2} \and Ruy J. G. B. de Queiroz\inst{2} \and Anjolina G. de Oliveira\inst{2} }

\institute{	Departamento de Matem\'atica\\
	Universidade Federal Rural de Pernambuco\\
	\email{tiago.veras@ufrpe.br}
	\and
	Centro de Inform\'atica\\
	Universidade Federal de Pernambuco\\
	\email{afr@cin.ufpe.br}\\
	\email{ruy@cin.ufpe.br}\\
	\email{ago@cin.ufpe.br}\\
}

\begin{document}

	\maketitle
	
	\begin{abstract}
 
     One of the most interesting entities of homotopy type theory is the identity type. It gives rise to an interesting interpretation of the equality, since one can semantically interpret the equality between two terms of the same type as a collection of homotopical paths between points of the same space. Since this is only a semantical interpretation, the addition of paths to the syntax of homotopy type theory has been recently proposed by De Queiroz, Ramos and De Oliveira \cite{Ruy1,Art3}. In these works, the authors propose an entity known as `computational path', proposed by De Queiroz and Gabbay in 1994 \cite{Ruy4}, and show that it can be used to formalize the identity type. We have found that it is possible to use these computational paths as a tool to achieve one central result of algebraic topology and homotopy type theory: the calculation of fundamental groups of surfaces. We review the concept of computational paths and the $LND_{EQ}$-$TRS$, which is a term rewriting system proposed by De Oliveira in 1994 \cite{Anjo1} to map redundancies between computational paths. We then proceed to calculate the fundamental group of the circle, cylinder, M{\"o}bius band, torus and the real projective plane. Moreover, we show that the use of computational paths make these calculations simple and straightforward, whereas the same result is much harder to obtain using the traditional code-encode-decode approach of homotopy type theory.  
	\end{abstract}
	
	\section{Introduction}\label{intro}
	
	One of the most interesting and intriguing concepts of Martin-L\"of's Type Theory is the identity type. It has been thoroughly studied since the discovery of the so-called Univalent Models by Vladimir Voevodsky in 2005 \cite{Vlad1}. This discovered gave rise to a new area of research known as homotopy type theory. It is based on a groundbreaking semantic interpretation of the identity type: a witness  $p: Id_{A}(a,b)$ can be seen as  a homotopical path between the points $a$ and $b$ within a topological space $A$. As one can see in \cite{hott,Steve1}, this interpretation generated groundbreaking results. Nonetheless, one has to assert the fact that this interpretation is a semantical one, i.e., those paths are not part of the syntax of type theory.
   
    Based on the fact that paths are seen only as semantic entities in homotopy type theory, De Queiroz, De Oliveira and Ramos have recently proposed a way of adding those paths to the syntax of type theory. They have been proposing the addition of an entity called \textbf{computational path} to type theory. This entity would work as the syntactic counterpart of semantic paths. Thus, one could interpret the identity type as syntactic paths between two terms of a given type. The full development of this theory can be seen in two recently published papers, \cite{Ruy1} and \cite{Art3}.
    
    Here the theory of computational paths is going to be thoroughly used. Our purpose is to further explore the study of fundamental groups of many surfaces using homotopy type theory. Nevertheless, we find that the main technique currently used in homotopy type theory to study fundamental groups looks too complex and seems to ask for more clarity. We are referring to the code-encode-decode technique, extensively used in many proofs of traditional homotopy type theory, as one can see in \cite{hott}. Since paths are not present in the syntax of traditional type theory, one needs to use this technique to simulate the path-space. Thus, we found in the theory of computational paths an effective way of calculating the fundamental groups without the use of what seems to be overly complex techniques such as the code-encode-decode one. With that in mind, we are going to review the main concepts of the theory of computational paths, as proposed in \cite{Ruy1,Art3}, but we shaill stop short of diving into more advanced concepts of this theory.
    
    Since the notion of fundamental groups of surfaces is one of the central topics of interest in algebraic topology, the main topic of interest of this work is the calculation of these fundamental groups by means of computational paths. Given a surface $S$, the fundamental group is obtained by studying the paths on the surface starting and ending at the some point $x_ {0} \in S$. These paths are called \textit {loops $_{x_0}$} and $x_ {0} \in S$ is the base point. But, as it will become clear later on, we are not interested in all kinds of loops, but only the ones that are not homotopic to the point $x_ {0}$ (since they are homotopically the same and this path will be denoted by $\rho_{x_{0}}$, meaning `the reflexivity path'). 
     	
    First we will calculate the fundamental group of the circle, $\mathbb{S}^{1}$. We will use its proof to obtain the fundamental group of the cylinder, $\mathbb{C}$, and that of the M{\"o}bius band, $\mathbb{K}^{2}$. Next, we will get the fundamental group of two much more complex surfaces, the Torus $\mathbb{T}^{2}$ and the real projective plane $\mathbb{P}^{2}$.

	\section{Computational Paths} \label{path}
	
Before we enter in details of computational paths, let's recall what motivated the introduction of computational paths to type theory. In type theory, our types are interpreted using the so-called Brower-Heyting-Kolmogorov Interpretation. That way, a semantic interpretation of types are not given by truth-values, but by the concept of proof as a primitive notion. Thus, we have \cite{Ruy1}:

\begin{tabbing}
	\textbf{a proof of the proposition:} \hbox{\ \ \ \ \ } \= \textbf{is given by:} \\
	$A\land B$ \> a proof of $A$ \textbf{and} a proof of $B$ \\ 
	$A\lor B$ \> a proof of $A$ \textbf{or} a proof of $B$ \\
	$A\rightarrow B$ \> a \textbf{function} that turns a proof of $A$ into a proof of $B$ \\
	$\forall x^D.P(x)$ \> a \textbf{function} that turns an element $a$ into a proof of $P(a)$ \\
	$\exists x^D.P(x)$ \> an element $a$ (witness) \textbf{and a proof of} $P(a)$ \\
\end{tabbing}

Also, based on the Curry-Howard functional interpretation of logical connectives, one have \cite{Ruy1}:

\begin{tabbing}
	\textbf{a proof of the proposition:} \hbox{\ \ \ \ \ } \= \textbf{has the canonical form of:} \\
	$A\land B$ \> $\langle p,q\rangle$ where $p$ is a proof of $A$ and $q$ is a proof of $B$ \\
	$A\lor B$ \> $i(p)$ where $p$ is a proof of $A$ or $j(q)$ where $q$ is a proof of $B$ \\
	\> (`$i$' and `$j$' abbreviate `into the left/right disjunct') \\
	$A\rightarrow B$ \> $\lambda x.b(x)$ where $b(p)$ is a proof of B \\
	\> provided $p$ is a proof of A \\
	$\forall x^A.B(x)$ \> $\Lambda x.f(x)$ where $f(a)$ is a proof of $B(a)$ \\ 
	\> provided $a$ is an arbitrary individual chosen\\
	\> from the domain $A$\\ 
	$\exists x^A.B(x)$ \> $\varepsilon x.(f(x),a)$ where $a$ is a witness\\
	\> from the domain $A$, $f(a)$ is a proof of $B(a)$ \\
\end{tabbing}

If one looks closely, there is one interpretation missing in the BHK-Interpretation. What constitutes a proof of $t_{1} = t_{2}$? In other words, what is a proof of an equality statement? We answer this by proposing that an equality between those two terms should be a sequence of rewrites starting from $t_{1}$ and ending at $t_{2}$. Thus, we would have \cite{Ruy1}:

\begin{tabbing}
	\textbf{a proof of the proposition:} \hbox{\ \ \ \ \ } \= \textbf{is given by:} \\
	\\
	$t_1= t_2$ \> ? \\
	\> (Perhaps a sequence of rewrites \\
	\> starting from $t_1$ and ending in $t_2$?) \\
\end{tabbing}

We call computational path the sequence of rewrites between these terms.

\subsection{Formal Definition}

Since computational path is a generic term, it is important to emphasize the fact that we are using the term computational path in the sense defined by \cite{Ruy5}. A computational path is based on the idea that it is possible to formally define when two computational objects $a,b : A$ are equal. These two objects are equal if one can reach $b$ from $a$ applying a sequence of axioms or rules. This sequence of operations forms a path. Since it is between two computational objects, it is said that this path is a computational one. Also, an application of an axiom or a rule transforms (or rewrite) an term in another. For that reason, a computational path is also known as a sequence of rewrites. Nevertheless, before we define formally a computational path, we can take a look at one famous equality theory, the $\lambda\beta\eta-equality$ \cite{lambda}:

\begin{definition}
	The \emph{$\lambda\beta\eta$-equality} is composed by the following axioms:
	
	\begin{enumerate}
		\item[$(\alpha)$] $\lambda x.M = \lambda y.M[y/x]$ \quad if $y \notin FV(M)$;
		\item[$(\beta)$] $(\lambda x.M)N = M[N/x]$;
		\item[$(\rho)$] $M = M$;
		\item[$(\eta)$] $(\lambda x.Mx) = M$ \quad $(x \notin FV(M))$.
	\end{enumerate}
	
	And the following rules of inference:

	\bigskip
	\noindent
	\begin{bprooftree}
		\AxiomC{$M = M'$ }
		\LeftLabel{$(\mu)$ \quad}
		\UnaryInfC{$NM = NM'$}
	\end{bprooftree}
	\begin{bprooftree}
		\AxiomC{$M = N$}
		\AxiomC{$N = P$}
		\LeftLabel{$(\tau)$}
		\BinaryInfC{$M = P$}
	\end{bprooftree}
	
	\bigskip
	\noindent
	\begin{bprooftree}
		\AxiomC{$M = M'$ }
		\LeftLabel{$(\nu)$ \quad}
		\UnaryInfC{$MN = M'N$}
	\end{bprooftree}
	\begin{bprooftree}
		\AxiomC{$M = N$}
		\LeftLabel{$(\sigma)$}
		\UnaryInfC{$N = M$}
	\end{bprooftree}
	
	\bigskip
	\noindent
	\begin{bprooftree}
		\AxiomC{$M = M'$ }
		\LeftLabel{$(\xi)$ \quad}
		\UnaryInfC{$\lambda x.M= \lambda x.M'$}
	\end{bprooftree}
	
	\bigskip
	
	
	
\end{definition}





\begin{definition} \cite{lambda}
	$P$ is $\beta$-equal or $\beta$-convertible to $Q$  (notation $P=_\beta Q$)
	iff $Q$ is obtained from $P$ by a finite (perhaps empty)  series of $\beta$-contractions
	and reversed $\beta$-contractions  and changes of bound variables.  That is,
	$P=_\beta Q$ iff \textbf{there exist} $P_0, \ldots, P_n$ ($n\geq 0$)  such that
	$P_0\equiv P$,  $P_n\equiv Q$,
	$(\forall i\leq n-1) (P_i\triangleright_{1\beta}P_{i+1}  \mbox{ or }P_{i+1}\triangleright_{1\beta}P_i  \mbox{ or } P_i\equiv_\alpha P_{i+1}).$
\end{definition}
\noindent (Note that equality has an \textbf{existential} force, which will show in the proof rules for the identity type.)

The same happens with $\lambda\beta\eta$-equality:\\
\begin{definition}($\lambda\beta\eta$-equality \cite{lambda})
	The equality-relation determined by the theory $\lambda\beta\eta$ is called
	$=_{\beta\eta}$; that is, we define
	$$M=_{\beta\eta}N\quad\Leftrightarrow\quad\lambda\beta\eta\vdash M=N.$$
\end{definition}

\begin{example}
	Take the term $M\equiv(\lambda x.(\lambda y.yx)(\lambda w.zw))v$. Then, it is $\beta\eta$-equal to $N\equiv zv$ because of the sequence:\\
	$(\lambda x.(\lambda y.yx)(\lambda w.zw))v, \quad  (\lambda x.(\lambda y.yx)z)v, \quad   (\lambda y.yv)z , \quad zv$\\
	which starts from $M$ and ends with $N$, and each member of the sequence is obtained via 1-step $\beta$- or $\eta$-contraction of a previous term in the sequence. To take this sequence into a {\em path\/}, one has to apply transitivity twice, as we do in the example below.
\end{example}

\begin{example}\label{examplepath}
	The term $M\equiv(\lambda x.(\lambda y.yx)(\lambda w.zw))v$ is $\beta\eta$-equal to $N\equiv zv$ because of the sequence:\\
	$(\lambda x.(\lambda y.yx)(\lambda w.zw))v, \quad  (\lambda x.(\lambda y.yx)z)v, \quad   (\lambda y.yv)z , \quad zv$\\
	Now, taking this sequence into a path leads us to the following:\\
	The first is equal to the second based on the grounds:\\
	$\eta((\lambda x.(\lambda y.yx)(\lambda w.zw))v,(\lambda x.(\lambda y.yx)z)v)$\\
	The second is equal to the third based on the grounds:\\
	$\beta((\lambda x.(\lambda y.yx)z)v,(\lambda y.yv)z)$\\
	Now, the first is equal to the third based on the grounds:\\
	$\tau(\eta((\lambda x.(\lambda y.yx)(\lambda w.zw))v,(\lambda x.(\lambda y.yx)z)v),\beta((\lambda x.(\lambda y.yx)z)v,(\lambda y.yv)z))$\\
	Now, the third is equal to the fourth one based on the grounds:\\
	$\beta((\lambda y.yv)z,zv)$\\
	Thus, the first one is equal to the fourth one based on the grounds:\\
	$\tau(\tau(\eta((\lambda x.(\lambda y.yx)(\lambda w.zw))v,(\lambda x.(\lambda y.yx)z)v),\beta((\lambda x.(\lambda y.yx)z)v,(\lambda y.yv)z)),\beta((\lambda y.yv)z,zv)))$.
\end{example}


The aforementioned theory establishes the equality between two $\lambda$-terms. Since we are working with computational objects as terms of a type, we need to translate the $\lambda\beta\eta$-equality to a suitable equality theory based on Martin L\"of's type theory. We obtain:

\begin{definition}
	The equality theory of Martin L\"of's type theory has the following basic proof rules for the $\Pi$-type:
	
	\bigskip
	
	\noindent
	\begin{bprooftree}
		\hskip -0.3pt
		\alwaysNoLine
		\AxiomC{$N : A$}
		\AxiomC{$[x : A]$}
		\UnaryInfC{$M : B$}
		\alwaysSingleLine
		\LeftLabel{$(\beta$) \quad}
		\BinaryInfC{$(\lambda x.M)N = M[N/x] : B[N/x]$}
	\end{bprooftree}
	\begin{bprooftree}
		\hskip 11pt
		\alwaysNoLine
		\AxiomC{$[x : A]$}
		\UnaryInfC{$M = M' : B$}
		\alwaysSingleLine
		\LeftLabel{$(\xi)$ \quad}
		\UnaryInfC{$\lambda x.M = \lambda x.M' : (\Pi x : A)B$}
	\end{bprooftree}
	
	\bigskip
	
	\noindent
	\begin{bprooftree}
		\hskip -0.5pt
		\AxiomC{$M : A$}
		\LeftLabel{$(\rho)$ \quad}
		\UnaryInfC{$M = M : A$}
	\end{bprooftree}
	\begin{bprooftree}
		\hskip 100pt
		\AxiomC{$M = M' : A$}
		\AxiomC{$N : (\Pi x : A)B$}
		\LeftLabel{$(\mu)$ \quad}
		\BinaryInfC{$NM = NM' : B[M/x]$}
	\end{bprooftree}
	
	\bigskip
	
	\noindent
	\begin{bprooftree}
		\hskip -0.5pt
		\AxiomC{$M = N : A$}
		\LeftLabel{$(\sigma) \quad$}
		\UnaryInfC{$N = M : A$}
	\end{bprooftree}
	\begin{bprooftree}
		\hskip 105pt
		\AxiomC{$N : A$}
		\AxiomC{$M = M' : (\Pi x : A)B$}
		\LeftLabel{$(\nu)$ \quad}
		\BinaryInfC{$MN = M'N : B[N/x]$}
	\end{bprooftree}
	
	\bigskip
	
	\noindent
	\begin{bprooftree}
		\hskip -0.5pt
		\AxiomC{$M = N : A$}
		\AxiomC{$N = P : A$}
		\LeftLabel{$(\tau)$ \quad}
		\BinaryInfC{$M = P : A$}
	\end{bprooftree}
	
	\bigskip
	
	\noindent
	\begin{bprooftree}
		\hskip -0.5pt
		\AxiomC{$M: (\Pi x : A)B$}
		\LeftLabel{$(\eta)$ \quad}
		\RightLabel {$(x \notin FV(M))$}
		\UnaryInfC{$(\lambda x.Mx) = M: (\Pi x : A)B$}
	\end{bprooftree}
	
	\bigskip
	
\end{definition}

We are finally able to formally define computational paths:

\begin{definition}
	Let $a$ and $b$ be elements of a type $A$. Then, a \emph{computational path} $s$ from $a$ to $b$ is a composition of rewrites (each rewrite is an application of the inference rules of the equality theory of type theory or is a change of bound variables). We denote that by $a =_{s} b$.
\end{definition}

As we have seen in \emph{example \ref{examplepath}}, composition of rewrites are applications of the rule $\tau$. Since change of bound variables is possible, each term is considered up to $\alpha$-equivalence.

\subsection{Equality Equations}

One can use the aforementioned axioms to show that computational paths establishes the three fundamental equations of equality: the reflexivity, symmetry and transitivity:

\bigskip

\begin{bprooftree}
	\AxiomC{$a =_{t} b : A$}
	\AxiomC{$b =_{u} c : A$}
	\RightLabel{\IT{transitivity}}
	\BinaryInfC{$a =_{\tau(t,u)} c : A$}
\end{bprooftree}
\begin{bprooftree}
	\AxiomC{$a : A$}
	\RightLabel{\IT{reflexivity}}
	\UnaryInfC{$a =_{\rho} a : A$}
\end{bprooftree}

\bigskip
\begin{bprooftree}
	\AxiomC{$a =_{t} b : A$}
	\RightLabel{\IT{symmetry}}
	\UnaryInfC{$b =_{\sigma(t)} a : A$}
\end{bprooftree}

\bigskip

\subsection{Identity Type}

We have said that it is possible to formulate the identity type using computational paths. As we have seen, the best way to define any formal entity of type theory is by a set of natural deductions rules. Thus, we define our path-based approach as the following set of rules: 

\begin{itemize}
	
	\item Formation and Introduction rule \cite{Ruy1,Art3}:
	
	\bigskip
	\begin{center}
		\begin{bprooftree}
			\AxiomC{$A$ type}
			\AxiomC{$a : A$}
			\AxiomC{$b : A$}
			\RightLabel{$Id - F$}
			\TrinaryInfC{$Id_{A}(a,b)$ type}
		\end{bprooftree}
		
		\bigskip
		
		\begin{bprooftree}
			\AxiomC{$a =_{s} b : A$}
			\RightLabel{$Id - I$}
			\UnaryInfC{$s(a,b) : Id_{A}(a,b)$}
		\end{bprooftree}
	\end{center}
	\bigskip
	
	One can notice that our formation rule is exactly equal to the traditional identity type. From terms $a, b : A$, one can form that is inhabited only if there is a proof of equality between those terms, i.e., $Id_{A}(a,b)$.
	
	The difference starts with the introduction rule. In our approach, one can notice that we do not use a reflexive constructor $r$. In other words, the reflexive path is not the main building block of our identity type. Instead, if we have a computational path $a =_{s} b : A$, we introduce $s(a,b)$ as a term of the identity type. That way, one should see $s(a,b)$ as a sequence of rewrites and substitutions (i.e., a computational path) which would have started from $a$ and arrived at $b$
	
	\bigskip 
	
	\item Elimination rule \cite{Ruy1,Art3}:
	
	\begin{center}
		\begin{bprooftree}
			\alwaysNoLine
			\AxiomC{$m : Id_{A}(a,b)$ }
			\AxiomC{$[a =_{g} b : A]$}
			\UnaryInfC{$h(g) : C$}
			\alwaysSingleLine
			\RightLabel{$Id - E$}
			\BinaryInfC{$REWR(m, \acute{g}.h(g)) : C$}
		\end{bprooftree}
	\end{center}
	\bigskip
	
	Let's recall the notation being used. First, one should see $h(g)$ as a functional expression $h$ which depends on $g$. Also, one should notice the the use of `$\acute{\ }$' in $\acute{g}$. One should see `$\acute{\ }$' as an abstractor that binds the occurrences of the variable $g$ introduced in the local assumption $[a =_{g} b : A]$ as a kind of {\em Skolem-type\/} constant denoting the {\em reason\/} why $a$ was assumed to be equal to $b$.
	
	We also introduce the constructor $REWR$. In a sense, it is similar to the constructor $J$ of the traditional approach, since both arise from the elimination rule of the identity type. The behavior of $REWR$ is simple. If from a computational path $g$ that establishes the equality between $a$ and $b$ one can construct $h(g) : C$, then if we also have this equality established by a term $m : Id_{A}(a,b)$, we can put together all this information in $REWR$ to construct $C$, eliminating the type $Id_{A}(a,b)$ in the process. The idea is that we can substitute $g$ for $m$ in $\acute{g}.h(g)$, resulting in $h(m/g) : C$. This behavior is established next by the reduction rule.
	
	\item Reduction rule \cite{Ruy1,Art3}:
	
	\bigskip
	\begin{center}
		\begin{bprooftree}
			\AxiomC{$a =_{m} b : A$}
			\RightLabel{$Id - I$}
			\UnaryInfC{$m(a,b) : Id_{A}(a,b)$}
			\alwaysNoLine
			\AxiomC{$[a =_{g} b : A]$}
			\UnaryInfC{$h(g) : C$}
			\alwaysSingleLine
			\RightLabel{$Id - E$ \quad $\rhd_\beta$}
			\BinaryInfC{$REWR(m, \acute{g}.h(g)) : C$}
		\end{bprooftree}
		\begin{bprooftree}
			\AxiomC{$[a =_{m} b : A]$}
			\alwaysNoLine
			\UnaryInfC{$h(m/g):C$}
		\end{bprooftree}
	\end{center}
	\bigskip
	
	\item Induction rule:
	
	\bigskip
	\begin{center}
		\begin{bprooftree}
			\AxiomC{$e : Id_{A}(a,b)$}
			\AxiomC{$[a =_{t} b : A]$}
			\RightLabel{$Id - I$}
			\UnaryInfC{$t(a, b) : Id_{A}(a, b)$}
			\RightLabel{$Id - E$ \quad  $\rhd_{\eta}$ \quad $e : Id_{A}(a,b)$}
			\BinaryInfC{$REWR(e, \acute{t}.t(a,b)) : Id_{A}(a,b)$}
		\end{bprooftree}
	\end{center}
	\bigskip
	
\end{itemize}

Our introduction and elimination rules reassures the concept of equality as an \BF{existential force}. In the introduction rule, we encapsulate the idea that a witness of a identity type $Id_{A}(a,b)$ only exists if there exist a computational path establishing the equality of $a$ and $b$. Also, one can notice that our elimination rule is similar to the elimination rule of the existential quantifier.

\subsection{Path-based Examples}

The objective of this subsection is to show how to use in practice the rules that we have just defined. The idea is to show construction of terms of some important types. The constructions that we have chosen to build are the reflexive, transitive and symmetric type of the identity type. Those were not random choices. The main reason is the fact that reflexive, transitive and symmetric types are essential to the process of building a groupoid model for the identity type \cite{hofmann1}. As we shall see, these constructions come naturally from simple computational paths constructed by the application of axioms of the equality of type theory.

Before we start the constructions, we think that it is essential to understand how to use the eliminations rules. The process of building a term of some type is a matter of finding the right reason. In the case of $J$, the reason is the correct $x,y : A$ and $z : Id_{A}(a,b)$ that generates the adequate $C(x,y,z)$. In our approach, the reason is the correct path $a =_{g} b$ that generates the adequate $g(a,b) : Id(a,b)$.

\subsubsection{Reflexivity.}

One could find strange the fact that we need to prove the reflexivity. Nevertheless, just remember that our approach is not based on the idea that reflexivity is the base of the identity type. As usual in type theory, a proof of something comes down to a construction of a term of a type. In this case, we need to construct a term of type $\Pi_{(a : A)}Id_{A}(a,a)$. The reason is extremely simple: from a term $a : A$, we obtain the computational path $a =_{\rho} a : A$ \cite{Art3}:

\bigskip

\begin{center}
	\begin{bprooftree}
		\AxiomC{$[a : A]$}
		\UnaryInfC{$a =_{\rho} a : A$}
		\RightLabel{$Id - I$}
		\UnaryInfC{$\rho(a,a) : Id_{A}(a,a)$}
		\RightLabel{$\Pi-I$}
		\UnaryInfC{$\lambda a.\rho(a,a) : \Pi_{(a : A)}Id_{A}(a,a)$}
	\end{bprooftree}
\end{center}

\subsubsection{Symmetry.}

The second proposed construction is the symmetry. Our objective is to obtain a term of type  $\Pi_{(a : A)}\Pi_{(b : A)}(Id_{A}(a,b) \rightarrow Id_{A}(b,a))$.

We construct a proof using computational paths. As expected, we need to find a suitable reason. Starting from $a =_{t} b$, we could look at the axioms of \emph{definition 4.1} to plan our next step. One of those axioms makes the symmetry clear: the $\sigma$ axiom. If we apply $\sigma$, we will obtain $b =_{\sigma(t)} a$. From this, we can then infer that $Id_A$ is inhabited by $(\sigma(t))(b,a)$. Now, it is just a matter of applying the elimination \cite{Art3}:

\bigskip

\begin{center}
	\begin{bprooftree}
		\alwaysNoLine
		\AxiomC{$[a:A] \quad [b:A]$}
		\UnaryInfC{$[p(a,b) : Id_{A}(a,b)]$}
		\alwaysSingleLine
		\AxiomC{[$a =_{t} b : A$]}
		\UnaryInfC{$b =_{\sigma(t)} a : A$}
		\RightLabel{$Id - I$}
		\UnaryInfC{$(\sigma(t))(b,a) : Id_{A}(b,a)$}
		\RightLabel{$Id - E$}
		\BinaryInfC{$REWR(p(a,b),\acute{t}.(\sigma(t))(b,a)) : Id_{A}(b,a)$}
		\RightLabel{$\rightarrow - I$}
		\UnaryInfC{$\lambda p.REWR(p(a,b), \acute{t}.(\sigma(t))(b,a)) : Id_{A} (a,b) \rightarrow Id_{A}(b,a)$}
		\RightLabel{$\Pi-I$}
		\UnaryInfC{$\lambda b. \lambda p.REWR(p(a,b),\acute{t}.(\sigma(t))(b,a)) :  \Pi_{(b : A)}(Id_{A} (a,b) \rightarrow Id_{A}(b,a))$}
		\RightLabel{$\Pi-I$}
		\UnaryInfC{$\lambda a.\lambda b. \lambda p.REWR(p(a,b), \acute{t}.(\sigma(t))(b,a)) :  \Pi_{(a : A)}\Pi_{(b : A)}(Id_{A} (a,b) \rightarrow Id_{A}(b,a))$}
	\end{bprooftree}
\end{center}

\bigskip

\subsubsection{Transitivity.}
The third and last construction will be the transitivity. Our objective os to obtain a term of type  $\Pi_{(a : A)}\Pi_{(b : A)}\Pi_{(c : A)} (Id_{A}(a,b) \rightarrow Id_{A}(b,c) \rightarrow Id_{A}(a,c))$.

To build our path-based construction, the first step, as expected, is to find the reason. Since we are trying to construct the transitivity, it is natural to think that we should start with paths $a =_{t} b$ and $b =_{u} c$ and then, from these paths, we should conclude that there is a path $z$ that establishes that $a =_{z} c$. To obtain $z$, we could try to apply the axioms of \emph{definition 4.1}. Looking at the axioms, one is exactly what we want: the axiom $\tau$. If we apply $\tau$ to  $a =_{t} b$ and $b =_{u} c$, we will obtain a new path $\tau(t,u)$ such that $a = _{\tau(t,u)} c$. Using that construction as the reason, we obtain the following term \cite{Art3}:

\begin{center}
	\begin{figure}
		\begin{sideways}
			\begin{bprooftree}
				\alwaysNoLine
				\AxiomC{$[a:A] \quad [b:A]$}
				\UnaryInfC{$[w(a,b) : Id_{A}(a,b)]$}
				\alwaysNoLine
				\AxiomC{$[c:A]$}
				\UnaryInfC{$[s(b,c) : Id_{A}(b,c)]$}
				\alwaysSingleLine
				\AxiomC{$[a =_{t} b:A]$}
				\AxiomC{$[b =_{u} c:A]$}
				\BinaryInfC{$a =_{\tau(t,u)} c:A$}
				\RightLabel{$Id - I$}
				\UnaryInfC{$(\tau (t,u))(a,c) : Id_{A}(a,c)$}
				\RightLabel{$Id - E$}
				\BinaryInfC{$REWR(s(b,c),\acute{u}(\tau (t,u))(a,c)) : Id_{A}(a,c)$}
				\RightLabel{$Id - E$}
				\BinaryInfC{$REWR(w(a,b),\acute{t}REWR(s(b,c),\acute{u}(\tau (t,u))(a,c))) : Id_{A}(a,c)$}
				\RightLabel{$\rightarrow - I$}
				\UnaryInfC{$\lambda s.REWR(w(a,b),\acute{t}REWR(s(b,c),\acute{u}(\tau (t,u))(a,c))) : Id_{A}(b,c) \rightarrow Id_{A}(a,c)$}
				\RightLabel{$\rightarrow - I$}
				\UnaryInfC{$\lambda w.\lambda s.REWR(w(a,b),\acute{t}REWR(s(b,c),\acute{u}(\tau (t,u))(a,c))) : Id_{A}(a,b) \rightarrow Id_{A}(b,c) \rightarrow Id_{A}(a,c)$}
				\RightLabel{$\Pi-I$}
				\UnaryInfC{$\lambda c.\lambda w.\lambda s.REWR(w(a,b),\acute{t}REWR(s(b,c),\acute{u}(\tau (t,u))(a,c))) :  \Pi_{(c : A)}(Id_{A}(a,b) \rightarrow Id_{A}(b,c) \rightarrow Id_{A}(a,c))$}
				\RightLabel{$\Pi-I$}
				\UnaryInfC{$\lambda b. \lambda c.\lambda w.\lambda s.REWR(w(a,b),\acute{t}REWR(s(b,c),\acute{u}(\tau (t,u))(a,c))) :  \Pi_{(b : A)}\Pi_{(c : A)}(Id_{A}(a,b) \rightarrow Id_{A}(b,c) \rightarrow Id_{A}(a,c))$}
				\RightLabel{$\Pi-I$}
				\UnaryInfC{$\lambda a. \lambda b. \lambda c.\lambda w.\lambda s.REWR(w(a,b),\acute{t}REWR(s(b,c),\acute{u}(\tau (t,u))(a,c))) :   \Pi_{(a : A)}\Pi_{(b : A)}\Pi_{(c : A)}(Id_{A}(a,b) \rightarrow Id_{A}(b,c) \rightarrow Id_{A}(a,c))$}
			\end{bprooftree}
		\end{sideways}
	\end{figure}
\end{center}

\newpage

As one can see, each step is just straightforward applications of introduction, elimination rules and abstractions. The only idea behind this construction is just the simple fact that the axiom $\tau$ guarantees the transitivity of paths.

\section{A Term Rewriting System for Paths}
	
	As we have just shown, a computational path establishes when two terms of the same type are equal. From the theory of computational paths, an interesting case arises. Suppose we have a path $s$ that establishes that $a =_{s} b : A$ and a path $t$ that establishes that $a =_{t} b : A$. Consider that $s$ and $t$ are formed by distinct compositions of rewrites. Is it possible to conclude that there are cases that $s$ and $t$ should be considered equivalent? The answer is \emph{yes}. Consider the following example:
	
	\begin{example}
		\noindent \normalfont Consider the path  $a =_{t} b : A$. By the symmetry property, we obtain $b =_{\sigma(t)} a : A$. What if we apply the property again on the path $\sigma(t)$? We would obtain a path  $a =_{\sigma(\sigma(t))} b : A$. Since we applied symmetry twice in succession, we obtained a path that is equivalent to the initial path $t$. For that reason, we conclude the act of applying symmetry twice in succession is a redundancy. We say that the path $\sigma(\sigma(t))$ can be reduced to the path $t$.
	\end{example}
	
	As one could see in the aforementioned example, different paths should be considered equal if one is just a redundant form of the other. The example that we have just seen is just a straightforward and simple case. Since the equality theory has a total of 7 axioms, the possibility of combinations that could generate redundancies are rather high. Fortunately, most possible redundancies were thoroughly mapped by \cite{Anjo1}. In that work, a system that establishes redundancies and creates rules that solve them was proposed. This system, known as $LND_{EQ}$-$TRS$, originally mapped a total of 39 rules. For each rule, there is a proof tree that constructs it.  We included all rules in \textbf{appendix B}. To illustrate those rules, take the case of \textbf{example 2}. We have the following \cite{Ruy1}:
	
	\bigskip
	\begin{prooftree}
		\AxiomC{$x =_{t} y : A$}
		\UnaryInfC{$y =_{\sigma(t)} x : A$}
		\RightLabel{\quad $\rhd_{ss}$ \quad $x =_{t} y : A$}
		\UnaryInfC{$x =_{\sigma(\sigma(t))} y : A$}
	\end{prooftree}
	
	\bigskip
	
	It is important to notice that we assign a label to every rule. In the previous case, we assigned the label $ss$.

	\begin{definition}($rw$-rule \cite{Art3})
		\normalfont An $rw$-rule is any of the rules defined in $LND_{EQ}$-$TRS$.
	\end{definition}
	
	\begin{definition}($rw$-contraction \cite{Art3})
		Let $s$ and $t$ be computational paths. We say that $s \rhd_{1rw} t$ (read as: $s$ $rw$-contracts to $t$) iff we can obtain $t$ from $s$ by an application of only one $rw$-rule. If $s$ can be reduced to $t$ by finite number of $rw$-contractions, then we say that $s \rhd_{rw} t$ (read as $s$ $rw$-reduces to $t$).
		
	\end{definition}
	
	\begin{definition}($rw$-equality \cite{Art3})
		\normalfont  Let $s$ and $t$ be computational paths. We say that $s =_{rw} t$ (read as: $s$ is $rw$-equal to $t$) iff $t$ can be obtained from $s$ by a finite (perhaps empty) series of $rw$-contractions and reversed $rw$-contractions. In other words, $s =_{rw} t$ iff there exists a sequence $R_{0},....,R_{n}$, with $n \geq 0$, such that
		
		\centering $(\forall i \leq n - 1) (R_{i}\rhd_{1rw} R_{i+1}$ or $R_{i+1} \rhd_{1rw} R_{i})$
		
		\centering  $R_{0} \equiv s$, \quad $R_{n} \equiv t$
	\end{definition}
	
	\begin{proposition}\label{proposition3.7} $rw$-equality  is transitive, symmetric and reflexive.
	\end{proposition}
	
	\begin{proof}
		Comes directly from the fact that $rw$-equality is the transitive, reflexive and symmetric closure of $rw$.
	\end{proof}
	
	The above proposition is rather important, since sometimes we want to work with paths up to $rw$-equality. For example, we can take a path $s$ and use it as a representative of an equivalence class, denoting this by $[s]_{rw}$.
	
	We'd like to mention that  $LND_{EQ}$-$TRS$ is terminating and confluent. The proof of this can be found in \cite{Anjo1,Ruy2,Ruy3,RuyAnjolinaLivro}.
	
	One should refer to \cite{Ruy1,RuyAnjolinaLivro} for a more complete and detailed explanation of the rules of $LND_{EQ}$-$TRS$.

	
	\section{Fundamental Group of surfaces obtained by means of Computational Paths}
	
	The objective of this section is to show that it is possible to use computational paths to obtain the fundamental group of the some surfaces, and this is one of the main results of homotopy theory. We avoid again the use of the heavy and rather complicated machinery of the code-encode-decode approach. In what follows we will get the fundamental group of some surfaces.
	
	
	\subsection{Fundamental Group of Circle $S^{1}$}
	\begin{definition}[The circle $S^1$]
		The circle is the type generated by:
		
		\begin{itemize}
			\item A point - $base : S^1$	
			\item A computational path - $base \underset{loop}{=} base : S^1$.
		\end{itemize}
	\end{definition}
	
	The first thing one should notice is that this definition doest not use only the points of the type $S^1$, but also a computational path $loop$ between those points. That is why it is called a higher inductive type \cite{hott}. Our approach differs from the one developed in the HTT book on the fact that we do not need to simulate the path-space between those points, since computational paths do exist in the syntax of the theory. Thus, if one starts with a path  $base \underset{loop}{=} base : S^1$, one can naturally obtain additional paths applying the path-axioms $\rho$, $\tau$ and $\sigma$.  Thus, one has a path $\sigma(loop) = loop^{-1}$, $\tau(loop, loop)$, etc. In Martin-L\"of's type theory, the existence of those additional paths comes from establishing that the paths should be freely generated by the constructors \cite{hott}. In our approach, we do not have to appeal to this kind of argument, since all paths come naturally from direct applications of the axioms and the inference rules which define the theory of equality.
	
	With that in mind, one can define the fundamental group of a circle. In homotopy theory, the fundamental group is the one formed by all equivalence classes up to homotopy of paths (loops) starting from a point $a$ and also ending at $a$. Since the we use computational paths as the syntactic counterpart of homotopic paths in type theory, we use it to propose the following definition:
	
	\begin{definition}[$\Pi_{1}(A,a)$ structure]
		$\Pi_{1}(A,a)$ is a structure defined as follows:
		
		\begin{center}
			$\Pi_{1}(A, a) = \{[loop]_{rw} \mid a \underset{loop}{=} a: A\}$
		\end{center}
	\end{definition}
	
	We use this structure to define the fundamental group of a circle. We also need to show that it is indeed a group.
	
	\begin{proposition}
		$(\Pi_{1}(S,a), \circ)$ is a group.
	\end{proposition}
	
	\begin{proof}
		The first thing to define is the group operation $\circ$. Given any $a \underset{r}{=} a : S^1$ and $a \underset{t}{=} a : S^1$, we define $r \circ s$ as $\tau(s,r)$. Thus, we now need to check the group conditions:
		
		\begin{itemize}
			
			\item Closure: Given $a \underset{r}{=} a : S^1$ and $a \underset{t}{=} a : S^1$, $r \circ s$ must be a member of the group. Indeed, $r \circ s = \tau(s,r)$ is a computational path $a \underset{\tau(s,r)}{=} a : S^1$.
			\bigskip
			\item Inverse: Every member of the group must have an inverse. Indeed, if we have a path $r$, we can apply $\sigma(r)$. We claim that $\sigma(r)$ is the inverse of $r$, since we have:
			
			\begin{center}
				$\sigma(r) \circ r = \tau(r, \sigma(r)) \underset{tr}{=} \rho$
				
				$r \circ \sigma(r) = \tau(\sigma(r), r) \underset{tsr}{=} \rho$
			\end{center}
			
			Since we are working up to $rw$-equality, the equalities hold strictly.
			
			\item Identity: We use the path $a \underset{\rho}{=} a : S^1$ as the identity. Indeed, we have:
			
			\begin{center}
				$r \circ \rho = \tau(\rho,r) \underset{tlr}{=} r$
				
				$\rho \circ r = \tau(r,\rho) \underset{trr}{=} r$.
			\end{center}
			
			\item Associativity: Given any members of the group $a \underset{r}{=}a : S^1$, $a \underset{t}{=} a$ and $a \underset{s}{=} a$, we want that $r \circ (s \circ t) = (r \circ s) \circ t$:
			
			\begin{center}
				$r \circ (s \circ t) = \tau(\tau(t,s), r) \underset{tt}{=} \tau(t,\tau(s,r)) = (r \circ s) \circ t$
			\end{center}
			
		\end{itemize}
		
		All conditions have been satisfied. $(\Pi_{1}(S,a), \circ)$ is a group.
	\end{proof}
	
	Thus, 	$(\Pi_{1}(S,a), \circ)$ is indeed a group. We call this group the fundamental group of $S^1$. Therefore, the objective of this section is to show that $\Pi_{1}(S,a) \simeq \mathbb{Z}$.
	
	Before we start to develop this proof, the following lemma will prove to be useful:
	
	\begin{lemma}
		All paths generated by a path $a \underset{loop}{=} a$ are $rw$-equal to a path $loop^{n}$, for $n \in \mathbb Z$.
	\end{lemma}
	
	We have said that from a $loop$, one can freely generate different paths by applying composition $\tau$ and the symmetry $\sigma$. Thus, one can, for example, obtain something as $loop \circ loop \circ loop^{-1} \circ loop...$. Our objective with this lemma is to show that, in fact, this path can be reduced to a path of the form $loop^{n}$, for $n \in \mathbb Z$.
	
	\begin{proof}
		The idea is to proceed by induction. We start from a base $\rho$. For the base case, it is trivially true, since we define it to be equal to $loop^{0}$. From $\rho$, one can construct more complex paths by composing with $loop$ or $\sigma(loop)$ on each step.
		We have the following induction steps:
		
		\begin{itemize}
			\item A path of the form $\rho$ concatenated with $loop$: We have $\rho \circ loop = \tau(loop,\rho) \underset{trr}{=} loop = loop^{1}$;
			\bigskip
			\item A path of the form $\rho$ concatenated with $\sigma(loop)$: We have $\rho \circ \sigma(loop) = \tau(\sigma(loop),\rho) \underset{trr}{=}  \sigma(loop) = loop^{-1}$
			\bigskip
			\item A path of the form $loop^{n}$ concatenated with $loop$: We have $loop^{n} \circ loop = loop^{n+1}$.
			\bigskip
			\item A path of the form $loop^{n}$ concatenated with $\sigma(loop)$: We have $loop^{n} \circ \sigma(loop)$ $= (loop^{n-1} \circ loop) \circ \sigma(loop) \underset{tt}{=} loop^{n-1} \circ (loop \circ \sigma(loop)) =$ $loop^{n-1} \circ (\tau(\sigma(loop), loop)) \underset{tsr}{=} loop^{n-1} \circ \rho = \tau(\rho, loop^{n-1}) \underset{tlr}{=} loop^{n-1}$
			\bigskip
			\item A path of the form $loop^{-n}$ concatenated with $loop$: We have $loop^{-n}$ = $loop^{-(n - 1)} \circ loop^{-1} = loop^{-(n - 1)} \circ \sigma(loop)$. Thus, we have $(loop^{-(n - 1)} \circ \sigma(loop)) \circ loop$ $ \underset{tt}{=}$ $loop^{-(n - 1)} \circ (\sigma(loop) \circ loop)$ $=$ $loop^{-(n-1)} \circ \tau(loop,\sigma(loop)) \underset{tr}{=}$ $=$ $loop^{-(n-1)} \circ \rho = \tau(\rho,loop^{-(n-1)}) \underset{tlr}{=} loop^{-(n-1)}$.
			\bigskip
			\item a path of the form $loop^{-n}$ concatenated with $\sigma(loop)$: We have $loop^{-n} \circ loop^{-1} = loop^{-(n + 1)}$ 
		\end{itemize}
		
		Thus, every path is of the form $loop^{n}$, with $n \in \mathbb Z$.
	\end{proof}

	This lemma shows that every path of the fundamental group can be represented by a path of the form $loop^{n}$, with $n \in \mathbb Z$.
	
	\begin{theorem}
		$\Pi_{1}(S,a) \simeq \mathbb{Z}$
	\end{theorem}
	
	To prove this theorem, one could use the approach proposed in \cite{hott}, defining a pair of encode and decode functions. Nevertheless, since our computational paths are already part of the syntax, one does not need to rely on this kind of approach to simulate a path-space. We can work directly with the concept of path.
	
	\begin{proof}
		The proof is done by establishing a function from $\Pi_{1}(S,a)$ to $\mathbb{Z}$ and then an inverse from $\mathbb{Z}$ to $\Pi_{1}(S,a)$. Since we have access to the previous lemma, this task is not too difficult. The main idea is that the $n$ in $loop^{n}$ means the amount of times one goes around the circle, while the sign gives the direction (clockwise or anti-clockwise). In other words, it is the $winding$ number. Since we have shown that every path of the fundamental group is of the form $loop^{n}$, with $n \in \mathbb Z$, then we just need to translate $loop^{n}$ to an integer $n$ and an integer $n$ to a path $loop^{n}$. We define two functions, $toInteger: \Pi_{1}(S,a) \rightarrow \mathbb Z$ and $toPath: \mathbb Z \rightarrow \Pi_{1}(S,a)$:
		
		\begin{itemize}
			\item $toInteger$: To define this function, we use two functions defined in $\mathbb Z$: the successor function $succ$ and the predecessor function $pred$. We define $toInteger$ as follows. Of course, we use directly the fact that every path of $\Pi_{1}(S,a)$ is of the form $loop^{n}$ with $n \in \mathbb Z$:
			
			\begin{equation*}
			toInteger: \begin{cases}
			toInteger(loop^n \equiv \rho) = 0 \quad \quad \quad \quad \quad \quad \quad \quad \quad \quad \enskip n = 0          \\
			toInteger(loop^{n}) = succ(toInteger(loop^{n-1})) \quad \quad n > 0  \\
			toInteger(loop^{n}) = pred(toInteger(loop^{n+1})) \quad \quad n < 0 \\
			\end{cases}
			\end{equation*}
			
			\item $toPath$: We just need to transform an integer $n$ into a path $loop^{n}$:
			
			\begin{equation*}
			toPath: \begin{cases}
			toPath(n) = \rho \quad \quad \quad \quad \quad \quad \quad \quad \quad\quad \enskip n = 0 \\
			toPath(n) = toPath(n - 1) \circ loop \quad \quad n > 0 \\
			toPath(n) = toPath(n + 1) \circ \sigma(loop) \quad n < 0 \\
			\end{cases}
			\end{equation*}
		\end{itemize}
		
		That they are inverses is a straightforward check. Therefore, we have $\Pi_{1}(S,a) \simeq \mathbb{Z}$.
	\end{proof}
	
	The technique used to obtain the fundamental group of the circle will provide us with the means to obtain the fundamental group of the cylinder and the M{\"o}bius Band. In these cases, all proofs used in the circle can be used, thus we will focus on knowing what kinds of loops are in our interest and on getting the bijections. In the case of the torus and projective real plane, we will need more advanced analysis and tests.

	
	\subsection{Fundamental Group of Cylinder $\mathbb{C}$}
	
	We are interested in getting the fundamental group of a cylindrical surface $\mathbb{C}$. To do this, we need to choose \textit{loops$_{x_{0}}$} which are of interest to our study. All the \textit{loops$_{x_{0}}$} which do not go, at least, one full turn in the cylinder are homotopic to the point $x_{0}$, as shown in \textbf{figure 1}. When we write \textit{loops}  we refer to \textit{loops$_{x_{0}}$}.

\begin{figure}[!htb]
\centering
\includegraphics[width=0.4\columnwidth]{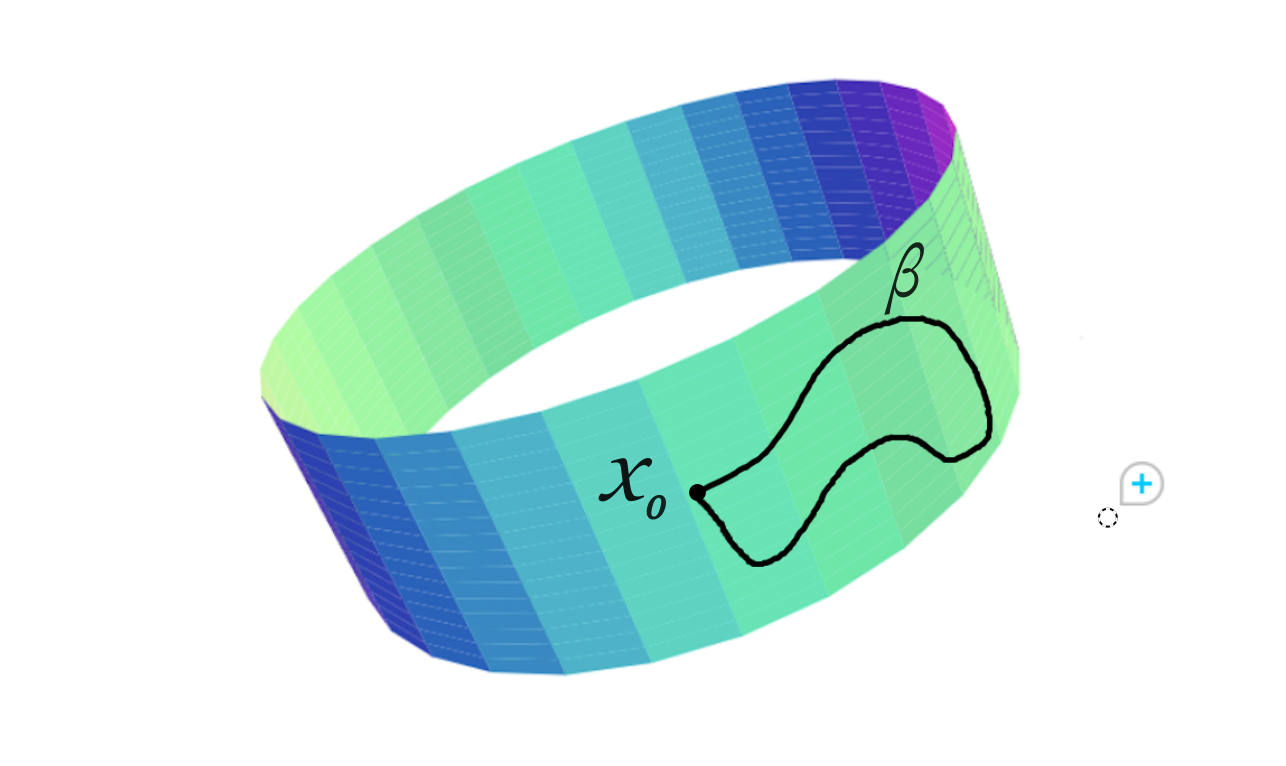}
\caption{ $\beta$ is a \textit{loop} in cylinder homotopic to $x_{0}$} 
\label{Rotulo4}
\end{figure}

Thus, loops that spin around the cylinder, such as the curve $\alpha$ in \textbf{figure 2}, cannot deform continuously to the point and therefore they are the loops of our interest.

\begin{figure}[!htb]
\centering
\includegraphics[width=0.4\columnwidth]{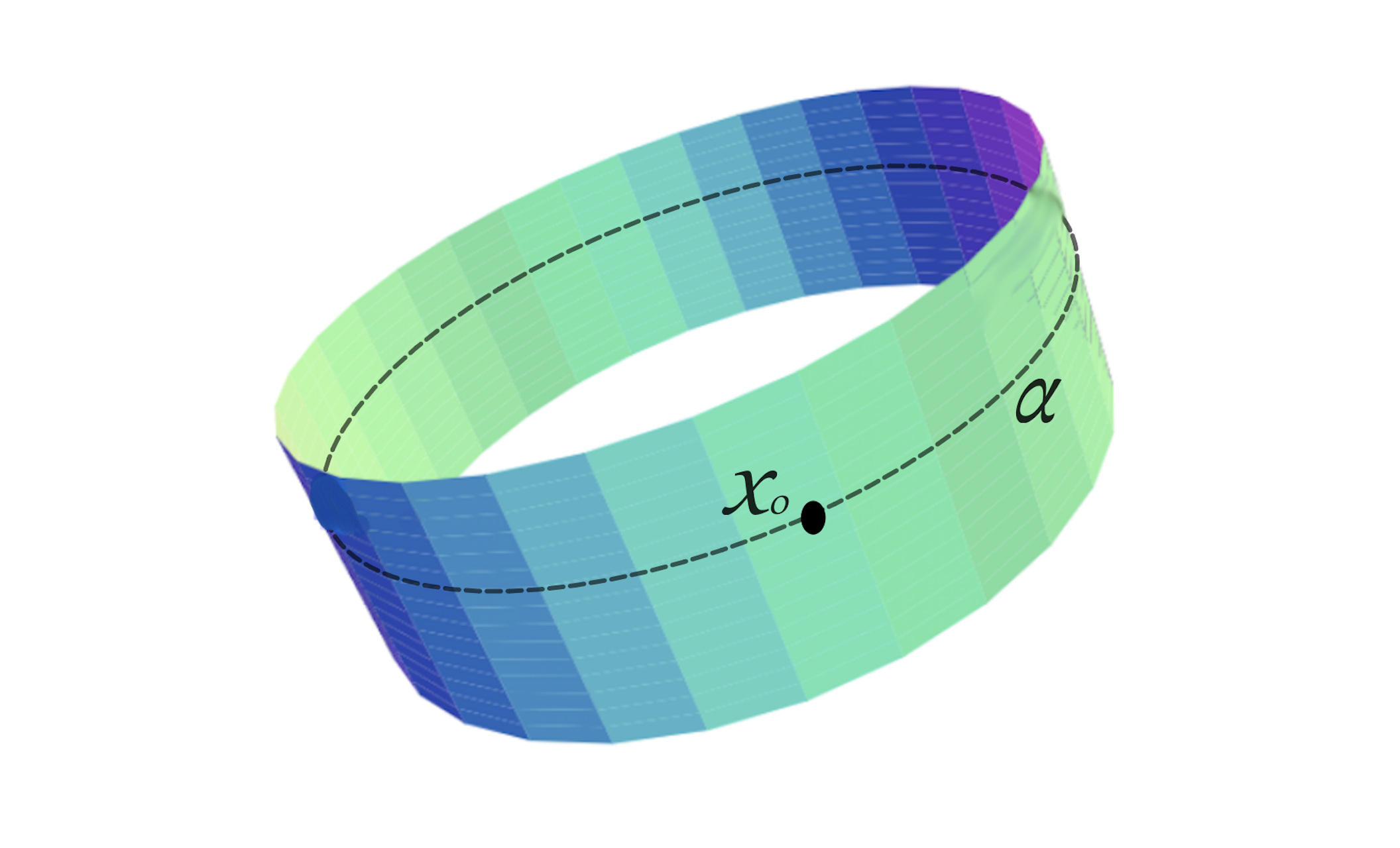}
\caption{$\alpha \in \Pi_{1}(\mathbb{C},x_{0})$} 
\label{Rotulo5}
\end{figure}

Similarly to the case of $\Pi_{1}(S^{1},x_{0})$, we can prove that all paths generated by $\alpha$ are $rw$-equal to a path \textit{loop$^{n}$} 
for a $n\in \mathbb{Z}$ and prove that $\Pi_{1}(\mathbb{C},x_{0})$ is a  group. Furthermore, using the maps $toPath$ and $toInteger$, we define the bijection between $\mathbb{Z}$ and $\Pi_{1}(\mathbb{C},x_{0})$.

\begin{theorem}
$\Pi_{1}\left(\mathbb{C},x_{0}\right) \simeq \mathbb{Z}.$
\end{theorem}

\begin{proof}

Proving this is equivalent to finding a bijection between spaces.


Consider the map:

\begin{eqnarray*}
toPath: \mathbb{Z} &\longrightarrow&  \Pi_{1}\left (\mathbb{C},x_{0}\right)\\
n  &\longrightarrow&  toPath(n).
\end{eqnarray*}

Dedined by:

$toPath(n) = \begin{cases}
toPath(0)=\rho &  \\
toPath(n)= toPath(n-1)\circ loop^{1}    \hspace{0.4cm}& n>0 \\
toPath(n)= toPath(n+1)\circ \sigma(loop^{1})    \hspace{0.4cm}& n<0. \\
\end{cases}$
$$$$
Now, consider the map:

\begin{eqnarray*}
toInteger:  \Pi_{1}\left (\mathbb{C},x_{0}\right)  &\longrightarrow& \mathbb{Z} \\
loop^{n}  &\longrightarrow&  n.
\end{eqnarray*}

Defined by:

$toInteger(n) = \begin{cases}

toInteger(loop^{0}=\rho)=0&  \\
toInteger(loop^{n})= succ(toInteger(loop^{n-1})) \hspace{0.4cm}& n>0 \\
toInteger(loop^{n})= pred(toInteger(loop^{n+1}))    \hspace{0.4cm}& n<0 \\
\end{cases}$
$$$$

That way, we have the desired isomorphism.
\end{proof}

\subsection{Fundamental Group of the M{\"o}bius Band - $\Pi_{1}(\mathbb{K}^{2},x_{0})$}	
	
In this case, we will again disregard all \textit{loops} which are homotopic to the constant $x_{0}$. Thus, the \textit{loops} of our interest are those that spin the surface in a fixed direction, such as the \textit{loops} denoted by \textit{$\alpha$} in \textbf{figure 3}.

\begin{figure}[!htb]
\centering
\includegraphics[width=0.4\columnwidth]{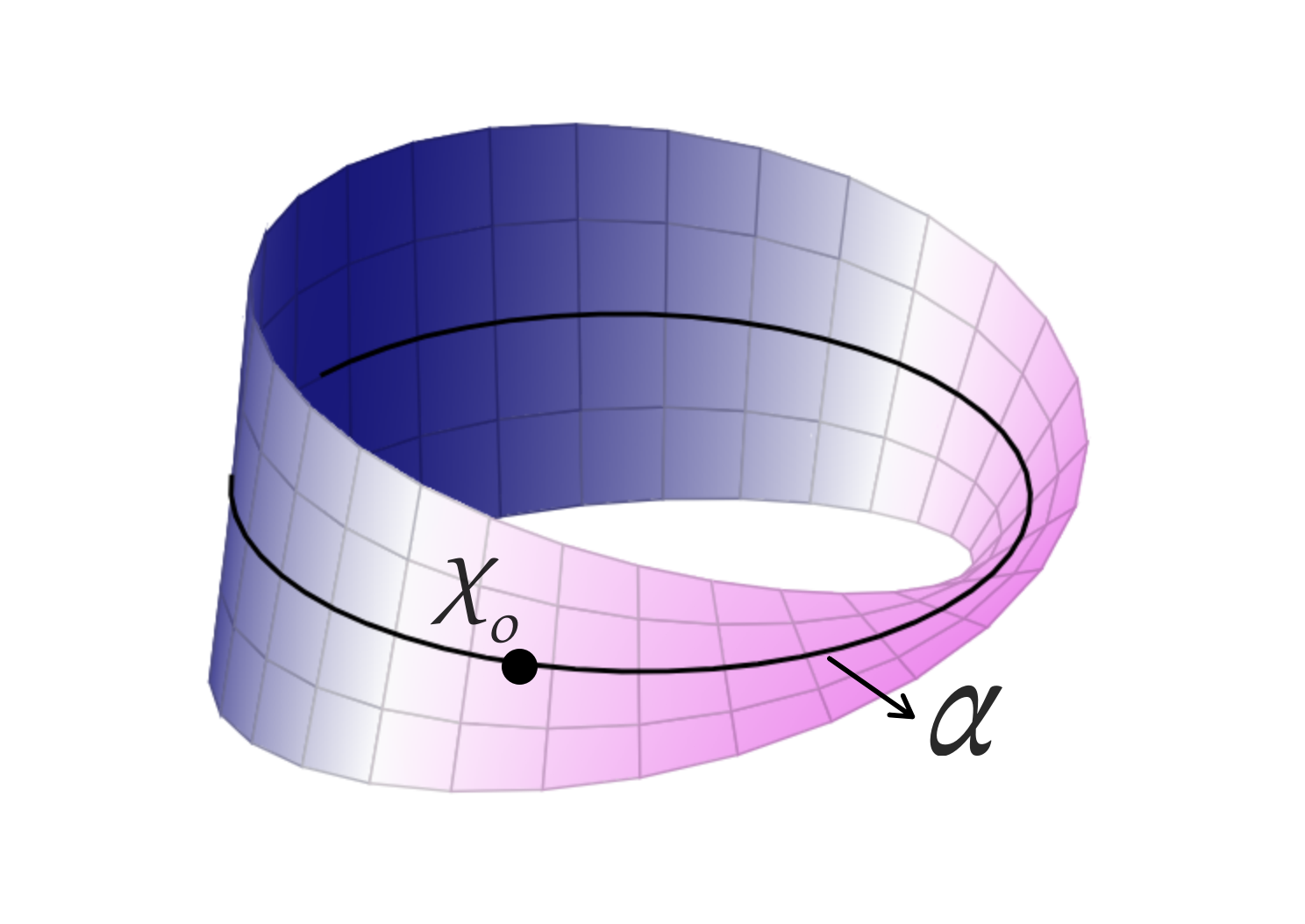}
\caption{$\alpha \in \Pi_{1}(\mathbb{K}^{2},x_{0})$} 
\label{Rotulo6}
\end{figure}

\begin{proposition}
$\Pi_{1}\left(\mathbb{K}^{2},x_{0}\right) \simeq \mathbb{Z}.$
\end{proposition}

\begin{proof}

Similarly to the case of $\Pi_{1}(\mathbb{S}^{1},x_{0})$, we can prove that all paths generated by $\alpha$ are $rw$-equal to a path \textit{loop$^{n}$} 
for a $n\in \mathbb{Z}$ and prove that $\Pi_{1}(\mathbb{K}^{2},x_{0})$ is a  group. Furthermore, using the maps $toPath$ and $toInteger$, we have the bijection between $\mathbb{Z}$ and $\Pi_{1}(\mathbb{K}^{2},x_{0})$. That way, we have the desired isomorphism.

\end{proof}


\subsection{Fundamental Group of the Torus - $\Pi_{1}(\mathbb{T}^{2},x_{0})$}
	
Consider $\mathbb{T}^{2}$ as the surface known as Torus and the point $x_{0}\in \mathbb{T}^{2}$. We will prove using computational paths that the fundamental group of the torus is isomorphic to $\mathbb{Z} \times \mathbb{Z}$.

Before we proceed, we need to look at some instances of \textit{loop$_{x_0}$} in $\mathbb{T}^{2}$. Consider the figure below:

\begin{figure}[!htb]
\centering
\includegraphics[width=0.3\columnwidth]{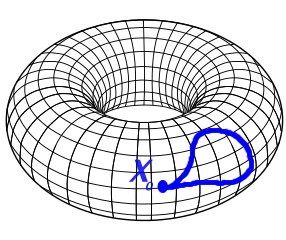}
\caption{ \textit{loop} homotopic to constant path $x_{0}$ in Torus.}
\label{Rotulo1}
\end{figure}

In \textbf{figure 4} above, we have an example of a \textit{loop$_{x_0}$} that is not particularly interesting because it continuously deforms to the constant path $\rho$. Thus, these types of loops will be discarded in our study.

Therefore, we will be interested in working with \textit{loops} that are not homotopic to the base point $x_{0}$, like \textit{loops} $\alpha$ and $\beta$. These loops will be the generators of $\mathbb{T}^{2}$ as shown in \textbf{figure 5}.

\begin{figure}[!htb]
\centering
\includegraphics[width=0.35\columnwidth]{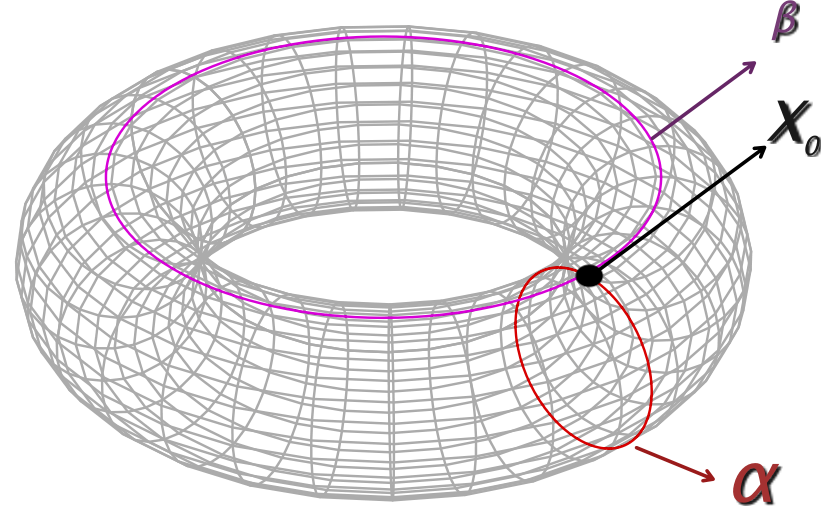}
\caption{Paths $\alpha$ and $\beta$ with base point $x_{0}$ in Torus} 
\label{Rotulo2}
\end{figure}

\begin{definition}[{vertical loop}]
   We define as vertical loop the path (loop) that passes through the inner part of $\mathbb{T}^{2}$ in the vertical direction. In \textbf{figure} 5, this loop is denoted by $\alpha$.
   
\end{definition}

\begin{definition}[{horizontal loop}]
We define as horizontal loop the path (loop) that passes the inner part of $\mathbb{T}^{2}$ in the horizontal direction. In \textbf{figure 5}, this \textit{loop} is denoted by $\beta$.

\end{definition}

Note that this two \textit{loops} are not of the type $\rho$ (homotopic to constant $x_{0}$). Furthermore, they generate $\mathbb{T}^{2}$. In what follows, we define and denote by: $\alpha^{n}=$\textit{loop$^{n}_{v}$} the path composed by $n$ vertical loops and by $\beta^{m}=$\textit{loop$^{m}_{h}$} the path composed by $m$ horizontal loops. 

We now give the formal definition of the torus in homotopy type theory:

\begin{definition}
 The torus $\mathbb{T}^{2}$ is generated by:
 
 \item[(i)] Two paths $\alpha$ and $\beta$ such that: $x_{0}\underset{\alpha}{=} x_{0}$ \hspace{0.2cm} and \hspace{0.2cm}  $x_{0}\underset{\beta}{=} x_{0}$.
 
 \item[(ii)] One path $co$ that establishes $\beta\alpha\underset{co}{=}\alpha\beta$, i.e., a term $co:Id(\beta\alpha,\alpha\beta)$.

\end{definition}

Given a point $x_{0}$, we can slice the Torus and represent it as a rectangle whose laterals are the loops $\alpha$ and $\beta$, how show in figure 6.

\begin{figure}[!htb]
\centering
\includegraphics[width=0.4\columnwidth]{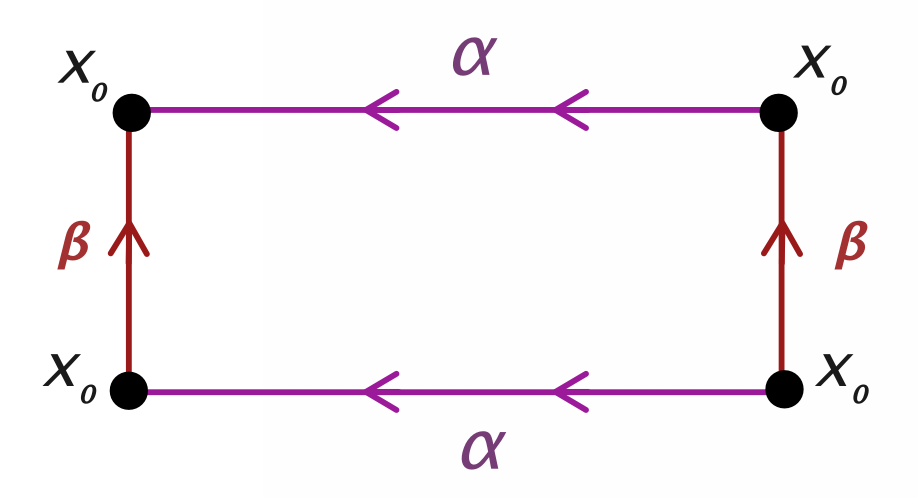}
\caption{Sliced Torus with oriented paths $\alpha$ and $\beta$ } 
\label{Rotulo3}
\end{figure}

Consider the following path in the figure: $\tau \left (\beta,\alpha,\sigma( \beta),\sigma( \alpha )\right)= \beta* \alpha * \beta^{-1} * \alpha^{-1}$:

\begin{proposition}
The aforementioned path is $rw$-equal to the reflexive path.
\end{proposition}

\begin{proof}
Indeed,

\begin{eqnarray*}
\beta*\alpha *\beta^{-1}*\alpha^{-1} &\underset{co}{=}& \beta*\beta^{-1} *\alpha * \alpha^{-1} \\
 &=& \tau \left(\tau(\sigma(\beta),\beta),\tau(\sigma(\alpha)\alpha,)\right)  \\
 &\underset{tst}{=}& \tau(\rho,\rho)\\
 &\underset{trr}{=}&\rho\\
\end{eqnarray*} 

 and thus,  $\tau \left (\beta,\alpha,\sigma( \beta),\sigma( \alpha )\right)= \beta* \alpha * \beta^{-1} * \alpha^{-1} = \rho$ \\
\end{proof}


 \begin{lemma}
 
All paths in $\mathbb{T}^{2}$ are \textit{rw-equal} to $\beta^{n} \alpha^{m}$, with $m,n \in \mathbb{Z}$ and $\beta^{0}, \alpha^{0}=\rho$.

\end{lemma}

\begin{proof}
Consider the following cases:

\item[(i)] Base case: $\beta^{0}\alpha^{0}=\rho$.
\item[(ii)] $\rho\circ \alpha=\tau(\alpha,\rho) \underset{trr}{=}\alpha =\beta^{0}\alpha^{1}.$
\item[(iii)] $\rho\circ \beta=\tau(\beta,\rho) \underset{trr}{=}\beta =\beta^{1}\alpha^{0}.$
\item[(iv)] $\rho\circ \alpha^{-1}=\tau(\sigma(\alpha),\rho) \underset{trr}{=}\sigma(\alpha) =\beta^{0}\alpha^{-1}.$
\item[(v)] $\rho\circ \beta^{-1}=\tau(\sigma(\beta),\rho) \underset{trr}{=}\sigma(\beta) =\beta^{-1}\alpha^{0}.$

Assuming, by the induction hypothesis, that every path is \textit{rw-equal}  $\beta^{n} \alpha^{m}$, we have:

\item[(1)] $\rho\circ \beta^{n}\alpha^{m}=\tau(\beta^{n}\alpha^{m},\rho) \underset{trr}{=}\beta^{n}\alpha^{m}.$
\item[(2)] $\alpha\circ \beta^{n}\alpha^{m}\underset{co}{=} \alpha\circ \alpha^{m}\beta^{n}=\alpha^{m+1}\beta^{n}\underset{co}{=}\beta^{n}\alpha^{m+1}.$
\item[(3)]  $\beta\circ \beta^{n}\alpha^{m}=\beta^{n+1}\alpha^{m} =\beta^{n}\alpha^{m+1}.$

\item[(4)] $\beta^{-1}\circ \beta^{n}\alpha^{m} =(\beta^{-1}\circ (\beta\circ \beta^{n-1}))\alpha^{m}\underset{tt}{=}((\beta^{-1}\circ \beta)\circ \beta^{n-1})\alpha^{m}\underset{tsr}{=}(\rho\circ\beta^{n-1})\alpha^{m}=\beta^{n-1}\alpha^{m}.$

\item[(4)] $\alpha^{-1}\circ \beta^{n}\alpha^{m}\underset{co}{=} \alpha^{-1}\circ \alpha^{m}\beta^{n}=(\alpha^{-1}\circ (\alpha\circ \alpha^{m-1}))\beta^{n}\underset{tt}{=} ((\alpha^{-1}\circ \alpha)\circ \alpha^{m-1})\beta^{n}\underset{tsr}{=}(\rho\circ\alpha^{m-1})\beta^{n}=\alpha^{m-1}\beta^{n} \underset{co}{=}\beta^{n}\alpha^{m-1}.$

So all paths in $\mathbb{T}^{2}$ are \textit{rw-equal} to $\beta^{n} \alpha^{m}.$
\end{proof}

\begin{proposition}
$ \left( \Pi_{1}(\mathbb{T}^{2},x_{0}),\circ \right)$ is a group.
\end{proposition}
\smallskip
\begin{proof}

\item [(+): Sum]  $$$$

\begin{prooftree}
\AxiomC{ $x_{0} \underset{\beta^{u}\alpha^{v}}{=}x_{0}$}
\AxiomC{$x_{0} \underset{\beta^{r}\alpha^{s}}{=} x_{0}$}
\BinaryInfC{ $x_{0} \underset{\tau \left(\beta^{u}\alpha^{v},\beta^{r}\alpha^{s}\right)}{=}x_{0}$}.
\end{prooftree}

But, 

\begin{eqnarray*}
\tau(\beta^{u}\alpha^{v},\beta^{u}\alpha^{v}) &=& (\beta^{r}\alpha^{s}) \circ (\beta^{u}\alpha^{v}) \\
 &=& \beta^{r}\alpha^{s} \beta^{u}\alpha^{v}\\
 &\underset{co}{=}& \beta^{r}\beta^{u}\alpha^{s} \alpha^{v}\\
 &=& \beta^{n}\alpha^{m} \in \Pi_{1}\left(T,x_{0}\right)
\end{eqnarray*}

\item[($\sigma$): Inverse]$$$$ 

\begin{prooftree}
\AxiomC{ $x_{0} \underset{\beta^{n}\alpha^{m}}{=}x_{0}$}
\AxiomC{$x_{0} \underset{\sigma{(\beta^{n})}\sigma{(\alpha^{m})}}{=} x_{0}$}
\BinaryInfC{ $x_{0} \underset{\tau \left(\beta^{n}\alpha^{m},\sigma{(\beta^{n})}\sigma{(\alpha^{m})}\right)}{=}x_{0}$}.
\end{prooftree}

But, 

\begin{eqnarray*}
\tau(\beta^{n}\alpha^{m},\sigma{(\beta^{n})}\sigma{(\alpha^{m})}) &=& (\sigma{(\beta^{n})}\sigma{(\alpha^{m})}) \circ (\beta^{n}\alpha^{m}) \\
 &=& \sigma{(\beta^{n})}\sigma{(\alpha^{m})} \beta^{n}\alpha^{m}\\
 &\underset{co}{=}&  \sigma{(\beta^{n})}\beta^{n}\sigma{(\alpha^{m})}\alpha^{m}\\
 &\underset{tsr}{=}&\rho_{\beta}\rho_{\alpha}\underset{trr}{=}\rho_{x_{0}}.
\end{eqnarray*}

On the other hand, we have:\\

\begin{prooftree}
\AxiomC{ $x_{0} \underset{\sigma{(\beta^{n})}\sigma{(\alpha^{m})}}{=}x_{0}$}
\AxiomC{$x_{0} \underset{\beta^{n}\alpha^{m}}{=} x_{0}$}
\BinaryInfC{ $x_{0} \underset{\tau \left(\sigma{(\beta^{n})}\sigma{(\alpha^{m}),\beta^{n}\alpha^{m}}\right)}{=}x_{0}$}.
\end{prooftree}

But, 

\begin{eqnarray*}
\tau(\sigma{(\beta^{n})}\sigma{(\alpha^{m})},\beta^{n}\alpha^{m}) &=& (\beta^{n}\alpha^{m}) \circ (\sigma{(\beta^{n})}\sigma{(\alpha^{m})})\\
 &=& \beta^{n}\alpha^{m}\sigma{(\beta^{n})}\sigma{(\alpha^{m})}\\
 &\underset{co}{=}& \beta^{n}\sigma{(\beta^{n})}\alpha^{m}\sigma{(\alpha^{m})}\\
 &\underset{tr}{=}&\rho_{\beta}\rho_{\alpha}\underset{trr}{=}\rho_{x_{0}}.
\end{eqnarray*}


\item[($\epsilon$): Identity]$$$$ 

\begin{prooftree}
\AxiomC{ $x_{0} \underset{\beta^{n}\alpha^{m}}{=}x_{0}$}
\AxiomC{$x_{0} \underset{\rho_{x_{0}}}{=} x_{0}$}
\BinaryInfC{ $x_{0} \underset{\tau \left(\beta^{n}\alpha^{m},\rho_{x_{0}}\right)}{=}x_{0}$}
\end{prooftree}

But, 

\begin{eqnarray*}
\tau(\beta^{n}\alpha^{m},\rho_{x_{0}}) &=& (\rho_{x_{0}}) \circ (\beta^{n}\alpha^{m})  \\
 &=& \rho_{x_{0}}\beta^{n}\alpha^{m}  \\
 &\underset{tlr}{=}&  \beta^{n}\alpha^{m}
\end{eqnarray*}
 
 and so \[
 \tau(\beta^{n}\alpha^{m},\rho_{x_{0}}) \underset{trr}{=} \beta^{n}\alpha^{m}.
\]

On the other hand, we have: $$$$

\begin{prooftree}
\AxiomC{ $x_{0} \underset{\rho_{x_{0}}}{=}x_{0}$}
\AxiomC{$x_{0} \underset{\beta^{n}\alpha^{m}}{=} x_{0}$}
\BinaryInfC{ $x_{0} \underset{\tau \left(\rho_{x_{0},\beta^{n}\alpha^{m}}\right)}{=}x_{0}$}.
\end{prooftree}

\bigskip
But, 

\begin{eqnarray*}
\tau \left(\rho_{x_{0}},\beta^{n}\alpha^{m}\right) &=&  (\beta^{n}\alpha^{m}) \circ (\rho_{x_{0}})  \\
 &=& \beta^{n}\alpha^{m} \rho_{x_{0}}  \\
 &\underset{trr}{=}&  \beta^{n}\alpha^{m}
\end{eqnarray*}

and so \[
 \tau(\rho_{x_{0}},\beta^{n}\alpha^{m}) \underset{trr}{=} \beta^{n}\alpha^{m}.
\]

\item[( $\circ$ ): Associativity]$$$$

\begin{prooftree}
\AxiomC{$x_{0} \underset{\beta^{n}\alpha^{m}}{=}x_{0}$}
\AxiomC{$x_{0} \underset{\beta^{i}\alpha^{j}}{=} x_{0}$}
\BinaryInfC{$ x_{0} \underset{\tau \left(\beta^{n}\alpha^{m},\beta^{i}\alpha^{j}\right)}{=}x_{0}$}
\AxiomC{$x_{0} \underset{\beta^{r}\alpha^{s}}{=}x_{0}$}
\BinaryInfC{$x_{0} \underset{\tau \left(\tau \left(\beta^{n}\alpha^{m},\beta^{i}\alpha^{j}\right),\beta^{r}\alpha^{s} \right)}{=}x_{0}$}.
\end{prooftree}

But, 

\begin{eqnarray*}
\tau \left(\tau \left(\beta^{n}\alpha^{m},\beta^{i}\alpha^{j}\right),\beta^{r}\alpha^{s} \right)  
 &=& (\beta^{r}\alpha^{s}) \circ \tau (\beta^{n}\alpha^{m},\beta^{i}\alpha^{j}) \\
 &=& (\beta^{r}\alpha^{s}) \circ (\beta^{i}\alpha^{j}\circ \beta^{n}\alpha^{m})\\
 &=& (\beta^{r}\alpha^{s}) \circ (\beta^{i}\alpha^{j}\beta^{n}\alpha^{m})\\
 &=& \beta^{r}\alpha^{s} \beta^{i}\alpha^{j}\beta^{n}\alpha^{m}.\\
\end{eqnarray*}
\bigskip

On the other hand, we have:

\begin{prooftree}
\AxiomC{$x_{0} \underset{\beta^{n}\alpha^{m}}{=}x_{0}$}
\AxiomC{$x_{0} \underset{\beta^{i}\alpha^{j}}{=} x_{0}$}
\AxiomC{$x_{0} \underset{\beta^{r}\alpha^{s}}{=} x_{0}$}
\BinaryInfC{$ x_{0} \underset{\tau \left(\beta^{i}\alpha^{j},\beta^{r}\alpha^{s}\right)}{=}x_{0}$}
\BinaryInfC{$x_{0} \underset{\tau (\beta^{n}\alpha^{m}, \tau \left(\beta^{i}\alpha^{j},\beta^{r}\alpha^{s})\right)}{=}x_{0}$}.
\end{prooftree}

But, 

\begin{eqnarray*}
\tau \left(\beta^{n}\alpha^{m} ,\tau \left(\beta^{i}\alpha^{j},\beta^{r}\alpha^{s}\right)\right)  
 &=&  \tau (\beta^{i}\alpha^{j},\beta^{r}\alpha^{s})\circ(\beta^{n}\alpha^{m}) \\
 &=& (\beta^{r}\alpha^{s} \circ \beta^{i}\alpha^{j})\circ (\beta^{n}\alpha^{m}) \\
 &=& (\beta^{r}\alpha^{s}\beta^{i}\alpha^{j})\circ (\beta^{n}\alpha^{m}) \\
 &=& \beta^{r}\alpha^{s} \beta^{i}\alpha^{j}\beta^{n}\alpha^{m}.\\
\end{eqnarray*}

Therefore, it follows that $ \left( \Pi_{1}(\mathbb{T}^{2},x_{0}),\circ \right)$ is a group.

\end{proof}

\begin{theorem}
$\Pi_{1}\left(\mathbb{T}^{2},x_{0}\right) \simeq \mathbb{Z} \times \mathbb{Z}.$
\end{theorem}

\begin{proof}

To prove this we need to find a bijection between spaces.

Consider the map:

\begin{eqnarray*}
toPath^{2}: \mathbb{Z} \times \mathbb{Z} &\longrightarrow&  \Pi_{1}\left (\mathbb{T}^{2},x_{0}\right)\\
(n,m)  &\longrightarrow&  toPath^{2}(n,m).
\end{eqnarray*}

Defined by:

$toPath^{2}(n,m) = \begin{cases}

toPath^{2}(0,0)=\rho &  \\
toPath^{2}(n,0)= toPath^{2}(n-1,0)\circ loop_{v}^{1}    \hspace{0.4cm}& n>0 \\
toPath^{2}(n,0)= toPath^{2}(n+1,0)\circ \sigma(loop_{v}^{1})    \hspace{0.4cm}& n<0 \\
toPath^{2}(n,m)= toPath^{2}(n,m-1)\circ (loop_{h}^{1})    \hspace{0.4cm}& m>0 \\
toPath^{2}(n,m)= toPath^{2}(n,m+1)\circ \sigma(loop_{h}^{1})    \hspace{0.4cm}& m<0 \\
\end{cases}$ \\ \\

Now, consider:

\begin{eqnarray*}
toInteger^{2}:  \Pi_{1}\left (\mathbb{T}^{2},x_{0}\right)  &\longrightarrow& \mathbb{Z} \times \mathbb{Z} \\
loop_{v}^{n}loop_{h}^{m}  &\longrightarrow&  (n,m).
\end{eqnarray*}

Defined by: 

$toInteger^{2}(n,m) = \begin{cases}

toInteger^{2}(\rho)=(0,0)&\\
toInteger^{2}(loop_{v}^{n})= toInteger^{2}(loop_{v}^{n-1})+(1,0) \hspace{0.4cm}& n>0\\
toInteger^{2}(loop_{v}^{n})= toInteger^{2}(loop_{v}^{n+1})+(-1,0)    \hspace{0.4cm}& n<0 \\
toInteger^{2}(loop_{v}^{n}loop_{h}^{m})= toInteger^{2}(loop_{v}^{n}loop_{h}^{m-1})+(0,1)    \hspace{0.4cm}& m>0 \\
toInteger^{2}(loop_{v}^{n}loop_{h}^{m})= toInteger^{2}(loop_{v}^{n}loop_{h}^{m+1})+(0,-1) \hspace{0.4cm}& m<0 \\
\end{cases}$\\

Therefore, we defined two injective maps: $$toInteger^{2}:  \Pi_{1}\left (\mathbb{T}^{2},x_{0}\right)  \longrightarrow \mathbb{Z} \times \mathbb{Z}$$ and $$toPath^{2}: \mathbb{Z} \times \mathbb{Z} \longrightarrow \Pi_{1}\left (\mathbb{T}^{2},x_{0}\right)$$ 

That way, we have that $\Pi_{1}\left(\mathbb{T}^{2},x_{0}\right)$ is isomorphic to $ \mathbb{Z} \times \mathbb{Z}$
\end{proof}

	
\subsection{Fundamental Group of the Real Projective Plane - $\Pi_{1}(\mathbb{P}^{2}, x_{0})$}

The real projective plane, denoted by $\mathbb{P}^{2}$, is by definition the set of all straight lines that pass through the origin of space $\mathbb{R}^3$. Each of these lines is a point in the projective plane. In $\mathbb{P}^{2}$, the points $(x,y,z)$ and $(\tilde{x},\tilde{y},\tilde{z}) \in \mathbb{R}^3-\{(0,0,0)\}$ are equivalents if, and only if, they are on the same line, that is, there is a constant of proportionality between them.

We can visualize $\mathbb{P}^2$ taking the sphere $\mathbb{S}^2$ of radius 1. at points in the sphere where $z\neq 0$ we can denote by $[x, y, z]$ the equivalents points $(x, y, z)$ and $(\tilde{x}, \tilde{y}, \tilde{z})$. Therefore, we can represent the projective plane as the upper hemisphere of the sphere together with the set of all pairs of antipodal point located on the curve $z =0$ in the sphere.

\begin{figure}[!htb]
\centering
\includegraphics[width=0.3\columnwidth]{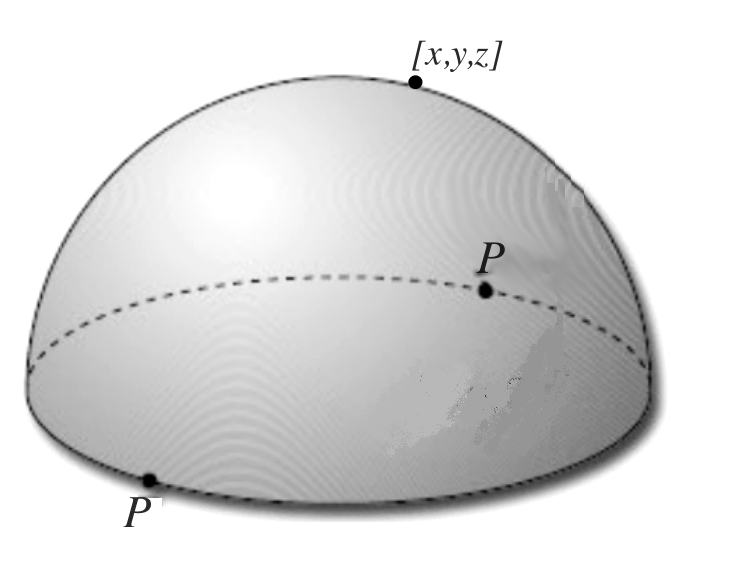}
\caption{Representation of $\mathbb{P}^2$ } 
\label{Rotulo7}
\end{figure}

Let's then map it on the unit disk through the following map $[x,y,z] \longrightarrow (x,y,0)$, as follows in the \textbf{figure 8}:

\begin{figure}[!htb]
\centering
\includegraphics[width=0.5 \columnwidth]{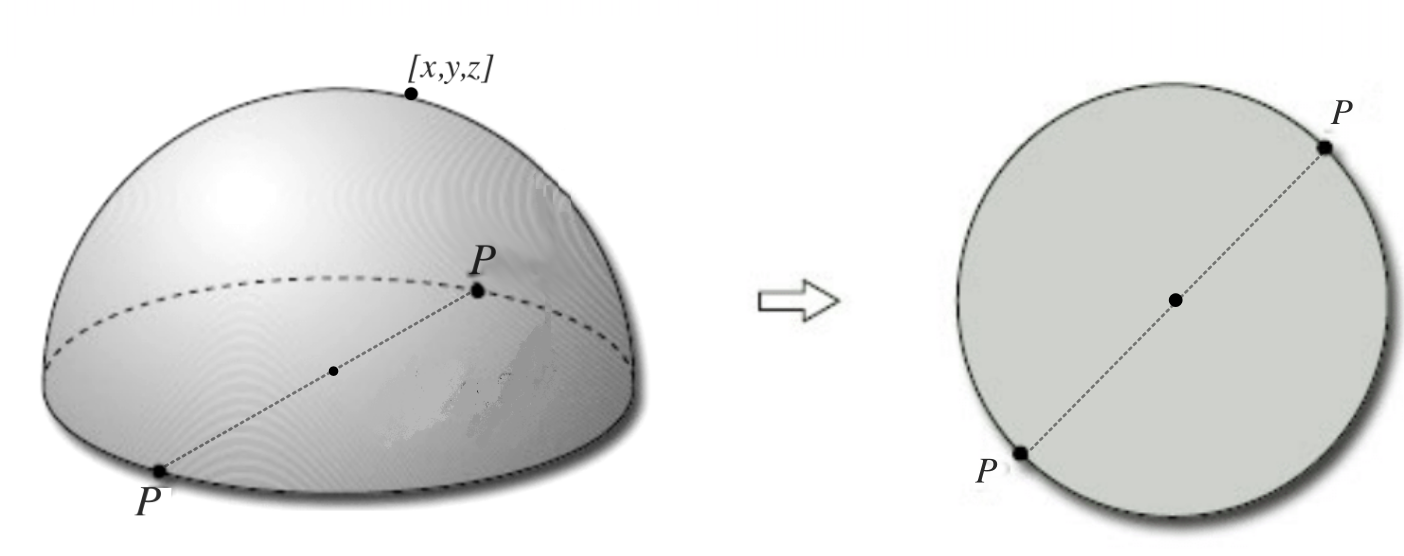}
\caption{Mapping Projection in the unit disk on $xy$ plane.} 
\label{Rotulo8}
\end{figure}

We denote by $\alpha$ any \textit{loop} that connects the identified antipodal points, so we can consider $\alpha$ as a \textit {loop}  (as follows in the \textbf{figure 9}) and any other \textit {loop} that connects the identified antipodal points is homotopic to $\alpha$.

\begin{figure}[!htb]
\centering
\includegraphics[width=0.6\columnwidth]{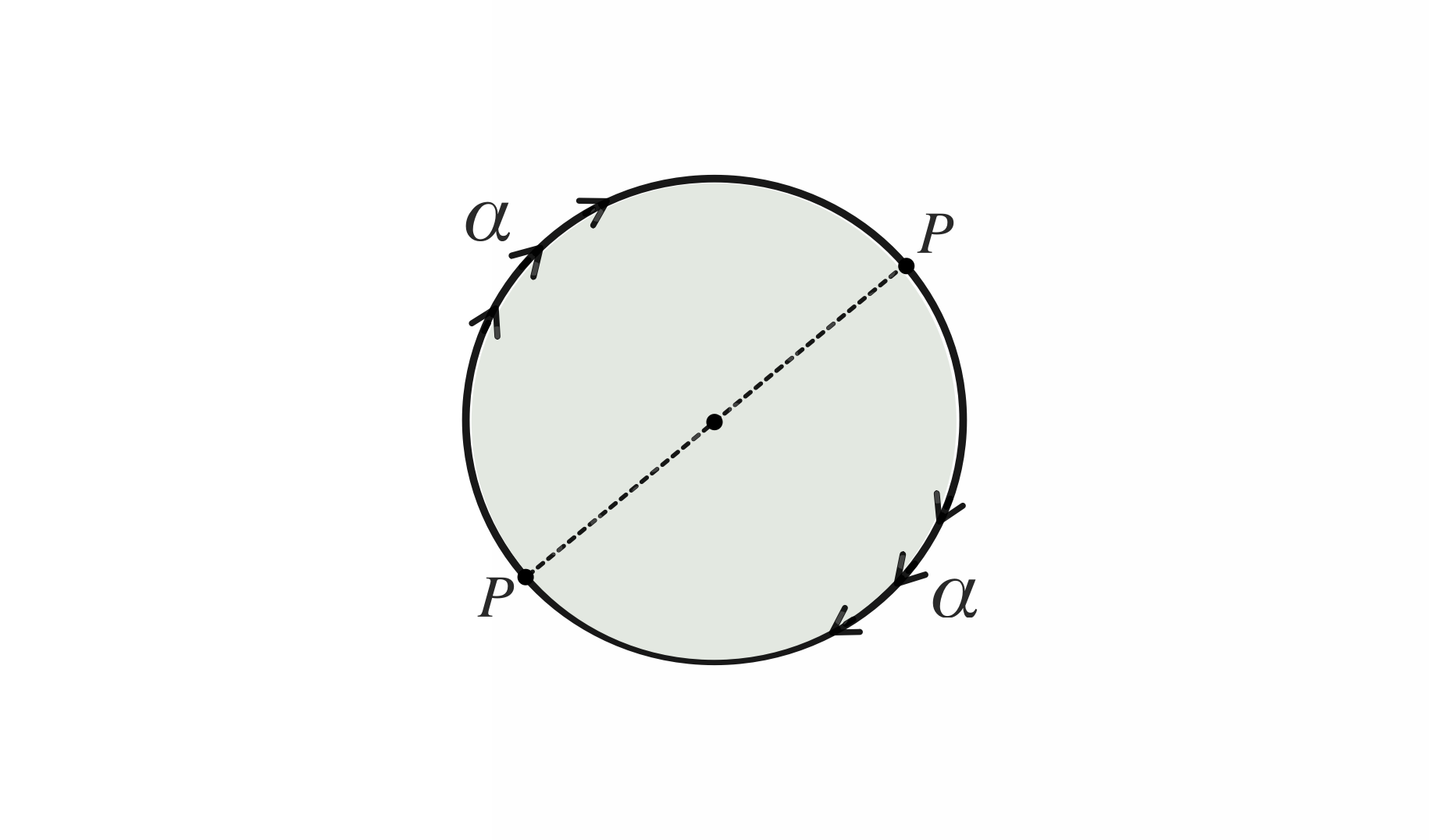}
\caption{\textit{loop} $\alpha$. } 
\label{Rotulo9}
\end{figure}

Since we can represent the real projective plane $\mathbb{P}^2$ for a disk $D_1$, we can define $\mathbb{P}^2$ it homotopically as follows:
\newline

\begin{definition}
 The real projective plane $\mathbb{P}^2$ is defined (inductively) by:
 
 \item[(i)] The points $Q:\mathbb{P}^2$, such that $Q\in Int D_1$, where $IntD_1$ is the set of all points in the interior of the disk.
 \item[(ii)] The pairs of antipodal points $P,P':\mathbb{P}^2$, such that $P,P'\in \partial D_1$, where $\partial D_1$ is the set of pairs of antipodal points.
 \item[(iii)]A path $\alpha:P=P'$.
 \item[(iv)] A path $cicl$ that establishes $\alpha \circ \alpha \underset{cicl}{=} \rho$, i.e,  $cicl:Id_{\mathbb{P}^2}(\alpha\circ \alpha,\rho)$.
 
\end{definition}

\begin{lemma}
 
All paths in $\mathbb{P}^{2}$ generated by $\rho$ or $\alpha$ are \textit{rw-equal} to $\rho$ or $\alpha$.
\end{lemma}

\begin{proof}

Consider the following base cases, $\rho$ and $\alpha$:

\item [Base case  $\rho$:]

\item[(i)] $\rho \circ \rho = \tau(\rho,\rho) \underset{trr}{=} \rho$.
\item[(ii)] $\rho \circ \alpha  = \tau(\alpha,\rho) \underset{trr}{=} \alpha$.
\item[(iii)] $\alpha \circ \rho = \tau(\rho,\alpha)  \underset{tlr}{=} \alpha$.

\item [Base case $\alpha$:]

\item[(i)] $\rho \circ \alpha = \tau(\alpha,\rho) \underset{trr}{=} \alpha$.
\item[(ii)] $\alpha \circ \rho  = \tau(\rho,\alpha) \underset{tlr}{=} \alpha$.
\item[(iii)] $\alpha \circ \alpha = \tau(\alpha,\alpha)  \underset{cicl}{=} \alpha$.

\item [Inductive case: Assuming true for $n$, we have:]

\item [If  \textit{loop$^{n}=\rho$}, we have two possibilities for $n+1$:]
\item[(i)] $loop^{n+1}=loop^{n}\circ\alpha = \rho \circ \rho =\tau(\rho,\rho)  \underset{trr}{=} \rho$.
\item[(ii)] $loop^{n+1}=loop^{n}\circ \alpha = \rho \circ \alpha =\tau(\alpha,\rho)  \underset{trr}{=} \alpha$.

\item [If  \textit{loop$^{n}=\alpha$}, we have two possibilities for $n+1$:]
\item[(i)] $loop^{n+1}=loop^{n}\circ\alpha = \alpha \circ \rho =\tau(\rho,\alpha)  \underset{tlr}{=} \alpha$.
\item[(ii)] $loop^{n+1}=loop^{n}\circ \alpha = \alpha \circ \alpha =\tau(\alpha,\alpha)  \underset{cicl}{=} \rho$.

\end{proof}

Thus, all paths in $\mathbb{P}^{2}$ generated by $\rho$ or $\alpha$ are \textit{rw-equal} to either $\alpha$ or $\rho$. Since we have $ \alpha \circ \alpha =\tau(\alpha,\alpha)  \underset{cicl}{=} \rho$, the term $cicl$ give us one important result: $\alpha=\sigma(\alpha)$.

\begin{proposition}
$\left( \Pi_{1}(\mathbb{P}^{2}),\circ \right)$ is a group.
\end{proposition}

\begin{proof}

\item [(+): Sum]$$$$

\begin{prooftree}
\AxiomC{ $P \underset{\alpha}{=}P$}
\AxiomC{$P \underset{\alpha}{=}P$}
\BinaryInfC { $P \underset{\tau \left(\alpha,\alpha\right)}{=}P$}
\end{prooftree}

But,  $$\alpha \circ \alpha = \tau \left(\alpha,\alpha\right) \underset{cicl}{=} \rho \in \Pi_{1}\left(\mathbb{P}^2\right).$$


\item[($\sigma$): Inverse]$$$$

\begin{prooftree}
\AxiomC{ $P \underset{\alpha}{=} P$}
\AxiomC{$P \underset{\sigma{(\alpha)}}{=} P$}
\BinaryInfC{ $P\underset{\tau \left(\alpha,\sigma{(\alpha)}\right)}{=}P$}
\end{prooftree}

But,  

$$\sigma(\alpha) \circ \alpha= \tau \left(\alpha,\sigma(\alpha)\right) \underset{tr}{=} \rho \in \Pi_{1}\left(\mathbb{P}^2\right).$$

On the other hand, we have:

\begin{prooftree}
\AxiomC{$P \underset{\sigma{(\alpha)}}{=} P$}
\AxiomC{ $P \underset{\alpha}{=} P$}
\BinaryInfC{ $P\underset{\tau \left(\sigma{(\alpha)}, \alpha \right)}{=}P$}
\end{prooftree}

But, 

 $$ \alpha \circ \sigma(\alpha) = \tau \left(\sigma(\alpha), \alpha\right) \underset{tsr}{=} \rho \in \Pi_{1}\left(\mathbb{P}^2\right).$$


\item[($\epsilon$): Identity]$$$$

\begin{prooftree}
\AxiomC{ $P \underset{\alpha}{=} P$}
\AxiomC{$P \underset{\rho}{=} P$}
\BinaryInfC{ $P \underset{\tau \left(\alpha,\rho\right)}{=} P$}
\end{prooftree}

But, 

 $$\rho \circ \alpha = \tau \left(\alpha, \rho\right) \underset{tlr}{=} \alpha \in \Pi_{1}\left(\mathbb{P}^2\right).$$

On the other hand, we have:

\begin{prooftree}
\AxiomC{$P \underset{\rho}{=} P$}
\AxiomC{ $P \underset{\alpha}{=} P$}
\BinaryInfC{ $P \underset{\tau \left(\rho, \alpha\right)}{=} P$}
\end{prooftree}

But, 

$$\alpha \circ \rho = \tau \left(\rho,\alpha\right) \underset{trr}{=} \alpha \in \Pi_{1}\left(\mathbb{P}^2\right).$$


\item[( $\circ$ ): Associativity]$$$$

\begin{prooftree}
\AxiomC{ $P \underset{\alpha}{=} P$}
\AxiomC{ $P \underset{\alpha}{=} P$}
\BinaryInfC{$ P \underset{\tau \left(\alpha,\alpha\right)}{=}P$}
\AxiomC{$P \underset{\alpha}{=}P$}
\BinaryInfC{$P \underset{\tau \left(\tau \left(\alpha,\alpha\right),\alpha \right)}{=}P$}
\end{prooftree}

But, 

\begin{eqnarray*}
\tau \left(\tau \left(\alpha,\alpha\right),\alpha \right)  
 &=& \alpha \circ \tau \left(\alpha,\alpha\right) \\
 &\underset{cicl}{=}& \alpha \circ \rho\\
 &=& \tau \left(\rho,\alpha\right)\\
 &\underset{trr}{=}& \alpha\\
\end{eqnarray*}
\bigskip

On the other hand, we have:

\begin{prooftree}
\AxiomC{$P \underset{\alpha}{=}P$}
\AxiomC{$P \underset{\alpha}{=} P$}
\AxiomC{$P \underset{\alpha}{=} P$}
\BinaryInfC{$ P \underset{\tau \left(\alpha,\alpha\right)}{=}P$}
\BinaryInfC{$P \underset{\tau (\alpha, \tau \left(\alpha,\alpha)\right)}{=}P$}
\end{prooftree}

But, 

\begin{eqnarray*}
\tau (\alpha, \tau \left(\alpha,\alpha)\right)  
 &=&  \tau (\alpha,\alpha)\circ\alpha \\
 &\underset{cicl}{=}& \rho \circ \alpha \\
 &=& \tau \left(\alpha, \rho\right) \\
&\underset{tlr}{=}& \alpha\\
\end{eqnarray*}

How  $\tau \left(\tau \left(\alpha,\alpha\right),\alpha \right)= \tau (\alpha, \tau \left(\alpha,\alpha)\right)$, fallow that associativity is valid and therefore $\left( \Pi_{1}(\mathbb{P}^{2}),\circ \right)$ is a group generated by $\rho$ and $\alpha$.

\end{proof}

\begin{theorem}
$ \Pi_{1}(\mathbb{P}^{2}) \simeq \mathbb{Z}_{2}$.
\end{theorem}

\begin{proof}

As before, proving this is equivalent to finding a bijection between spaces.$$$$

First, consider the map  $toPath_{\mathbb{Z}_2}:  \mathbb{Z}_2 \longrightarrow \Pi_{1}\left (\mathbb{P}^2\right)$ defined by:

$$toPath_{\mathbb{Z}_2} = \begin{cases}
toPath(0)=\rho &  \\
toPath(1)= \alpha = loop^{1}    \hspace{0.4cm}&  \\
\end{cases}$$ 
$$$$
Now, consider the map $toInteger: \Pi_{1}\left (\mathbb{P}^2\right) \longrightarrow \mathbb{Z}_2$  defined by:

$$toInteger = \begin{cases}
toInteger(loop^{0}=\rho)=0&  \\
toInteger(loop^{1}=\alpha)=1&\\
\end{cases}$$
$$$$

Thus, the isomorphism holds.
\end{proof}

	\section{Conclusion}
	
    In this work, our main objective has been the calculation of the fundamental groups of many surfaces using homotopy type theory. We have seen that it is possible to do these calculations by means of an entity known as computational paths. The main advantage of this approach is that we have avoided the use of more complex techniques, code-encode-decode one. As a consequence, our calculations proved to be straightforward and simple. Using computational paths as our main tool, we have calculated the fundamental group of the circle, cylinder, M{\"o}bius band, torus and projective plane. Therefore, we have shown that it is possible to use the theory of computational paths to obtain central results of algebraic topology and homotopy type theory.   
	
	\bibliographystyle{plain}
	\bibliography{ref1}	
	
	\newpage
	\appendix
	
	\section{Subterm Substitution}
	
	In Equational Logic, the sub-term substitution is given by the following inference rule \cite{Ruy2}:
	\begin{center}
		\begin{bprooftree}
			\AxiomC{$s = t$ }
			\UnaryInfC{$s\theta = t\theta$}
		\end{bprooftree}
	\end{center}
	
	One problem is that such rule does not respect the sub-formula property. To deal with that, \cite{chenadec} proposes two inference rules:
	
	\begin{center}
		\begin{bprooftree}
			\AxiomC{$M = N$}
			\AxiomC{$C[N] = O$}
			\RightLabel{$IL$ \quad}
			\BinaryInfC{$C[M] = O$}
		\end{bprooftree}
		\begin{bprooftree}
			\AxiomC{$M = C[N]$}
			\AxiomC{$N = O$}
			\RightLabel{$IR$ \quad}
			\BinaryInfC{$M = C[O]$}
		\end{bprooftree}
	\end{center}
	
	\noindent where M, N and O are terms.
	
	As proposed in \cite{Ruy1}, we can define similar rules using computational paths, as follows:
	
	\begin{center}
		\begin{bprooftree}
			\AxiomC{$x =_r {\cal C}[y]: A$}
			\AxiomC{$y =_s u : A'$}
			\BinaryInfC{$x =_{{\tt sub}_{\tt L}(r,s)} {\cal C}[u]: A$}
		\end{bprooftree}
		\begin{bprooftree}
			\AxiomC{$x =_r w : A'$}
			\AxiomC{${\cal C}[w]=_s u : A$}
			\BinaryInfC{${\cal C}[x]=_{{\tt sub}_{\tt R}(r,s)} u : A$}
		\end{bprooftree}
	\end{center}
	
	\noindent where $C$ is the context in which the sub-term detached by '[ ]' appears and $A'$ could be a sub-domain of $A$, equal to $A$ or disjoint to $A$.
	
	In the rule above, ${\cal C}[u]$ should be understood as the result of replacing every occurrence of $y$ by $u$ in $C$.
	
	\newpage
	
	\section{List of Rewrite Rules}
	
	We present all rewrite rules of $LND_{EQ}$-$TRS$. They are as follows (All have been taken from \cite{Ruy1}):
	\\
	
	\noindent 1. $\sigma(\rho) \triangleright_{sr} \rho$ \\ 
	2. $\sigma(\sigma(r)) \triangleright_{ss} r$\\ 
	3. $\tau({\cal C}[r] , {\cal C}[\sigma(r)]) \triangleright_{tr}  {\cal C }[\rho]$\\ 
	4. $\tau({\cal C}[\sigma(r)], {\cal C}[r]) \triangleright_{tsr} {\cal C}[\rho]$\\ 
	5. $\tau({\cal C}[r], {\cal C}[\rho]) \triangleright_{trr} {\cal C}[r]$\\ 
	6. $\tau({\cal C}[\rho], {\cal C}[r]) \triangleright_{tlr} {\cal C}[r]$ \\ 
	7. ${\tt sub_L}({\cal C}[r], {\cal C}[\rho]) \triangleright_{slr} {\cal C}[r]$\\ 
	8. ${\tt sub_R}({\cal C}[\rho], {\cal C}[r]) \triangleright_{srr} {\cal C}[r]$ \\
	9. ${\tt sub_L} ({\tt sub_L} (s, {\cal C}[r]), {\cal C}[\sigma(r)]) \triangleright_{sls} s$\\
	10. ${\tt sub_L} ( {\tt sub_L} (s , {\cal C}[\sigma(r)]) , {\cal C}[r]) \triangleright_{slss} s$\\ 
	11. ${\tt sub_R} ({\cal C}[s], {\tt sub_R} ({\cal C}[\sigma(s)],r)) \triangleright_{srs} r$\\ 
	12. ${\tt sub_R} ({\cal C}[\sigma(s)], {\tt sub_R} ({\cal C}[s] ,  r )) \triangleright_{srrr} r$\\ 
	13. 
	$\mu_1 ( \xi_1 ( r))\triangleright_{mx2l1} r$\\
	14. $\mu_1 ( \xi_\land ( r,s))\triangleright_{mx2l2} r$\\
	15.
	$\mu_2 ( \xi_\land ( r,s))\triangleright_{mx2r1} s$\\
	16.
	$\mu_2 ( \xi_2 ( s))\triangleright_{mx2r2} s$\\
	17. 
	$\mu ( \xi_1 (r) , s , u) \triangleright_{mx3l} s$\\ 
	18. 
	$\mu (\xi_2 (r) , s , u) \triangleright_{mx3r} u$\\ 
	19.
	$\nu (\xi (r)) \triangleright_{mxl} r$\\ 
	20.
	$\mu (\xi_2 (r) , s) \triangleright_{mxr} s$\\ 
	21.
	$\xi ( \mu_1 (r),\mu_2(r) ) \triangleright_{mx} r$ \\ 
	22.
	$\mu ( t, \xi_1 (r), \xi_2 (s)) \triangleright_{mxx} t$ \\ 
	23. 
	$\xi ( \nu (r) ) \triangleright_{xmr} r$ \\ 
	24. 
	$\mu (s,\xi_2 (r)) \triangleright_{mx1r} s$\\ 
	25. $\sigma(\tau(r,s)) \triangleright_{stss} \tau(\sigma(s),  \sigma(r))$\\ 
	26. $\sigma({\tt sub_L}(r,s)) \triangleright_{ssbl} {\tt sub_R}(\sigma(s), \sigma(r))$\\ 
	27. $\sigma ({\tt sub_R} (r,s)) \triangleright_{ssbr} {\tt sub_L} (\sigma
	(s),  \sigma (r))$\\ 
	28. $\sigma(\xi (r)) \triangleright_{sx} \xi ( \sigma(r))$\\ 
	29. $\sigma(\xi (s, r)) \triangleright_{sxss} \xi ( \sigma(s),  \sigma(r))$\\ 
	30. $\sigma(\mu (r)) \triangleright_{sm} \mu ( \sigma(r))$\\ 
	31. $\sigma(\mu (s, r)) \triangleright_{smss} \mu (\sigma(s),  \sigma(r))$\\ 
	32. $\sigma(\mu (r,u,v)) \triangleright_{smsss} \mu ( \sigma(r),\sigma(u),\sigma(v))$\\
	33. $\tau (r, {\tt sub_L} (\rho , s)) \triangleright_{tsbll} {\tt sub_L}  (r,s)$\\ 
	34. $\tau (r, {\tt sub_R} (s, \rho)) \triangleright_{tsbrl}  {\tt 
		sub_L} (r,s)$\\ 
	35. $\tau({\tt sub_L}(r,s),t) \triangleright_{tsblr} \tau (r, {\tt 
		sub_R} (s,t))$\\ 
	36. $\tau ({\tt sub_R} (s,t),u) \triangleright_{tsbrr} {\tt sub_R} (s, \tau  (t,u))$\\ 
	37. $\tau(\tau(t,r),s) \triangleright_{tt} \tau(t,\tau (r,s)) $\\
	38. $\tau ({\cal C}[u], \tau ({\cal C}[\sigma(u)] , v)) \triangleright_{tts} v$\\
	39. $\tau ({\cal C}[\sigma(u)] , \tau ({\cal C}[u] , v)) \triangleright_{tst} u$\\
	
\end{document}